\documentclass[11pt,a4paper]{article}

\usepackage[margin=1in]{geometry}
\usepackage{setspace}
\usepackage[T1]{fontenc}
\usepackage[utf8]{inputenc}
\usepackage{lmodern}

\usepackage{amsmath,amssymb,amsfonts,amsthm}
\usepackage{mathtools}
\usepackage{bbm}
\usepackage{bm}
\usepackage{subcaption}

\usepackage{graphicx}
\usepackage{tikz}
\usetikzlibrary{arrows.meta}
\usepackage{booktabs}
\usepackage{array}
\usepackage{multirow}

\usepackage{xcolor}
\usepackage[numbers,sort&compress]{natbib}
\usepackage[
    colorlinks=true,
    linkcolor=blue,
    citecolor=blue,
    urlcolor=blue
]{hyperref}
\usepackage[nameinlink,capitalise,noabbrev]{cleveref}
\usetikzlibrary{positioning}

\usepackage{enumitem}
\setlist[itemize]{leftmargin=2em}
\setlist[enumerate]{leftmargin=2em}

\theoremstyle{plain}
\newtheorem{theorem}{Theorem}[section]
\newtheorem{proposition}[theorem]{Proposition}
\newtheorem{lemma}[theorem]{Lemma}
\newtheorem{corollary}[theorem]{Corollary}

\theoremstyle{definition}
\newtheorem{definition}[theorem]{Definition}
\newtheorem{assumption}[theorem]{Assumption}
\newtheorem{example}[theorem]{Example}

\theoremstyle{remark}
\newtheorem{remark}[theorem]{Remark}

\title{A Mathematical Framework for Topological Causal Data Analysis}

\author{
Hugo Gobato Souto\thanks{Dell Technologies. Email: \texttt{hugo.souto@dell.com}}
\and
Ioannis Diamantis\thanks{Department of Data Analytics and Digitalisation, Maastricht University, Maastricht, The Netherlands. Email: \texttt{i.diamantis@maastrichtuniversity.nl}}
}

\date{}

\begin{document}

\maketitle

\begin{abstract}
Many modern outcomes, including images, point clouds, networks, and spatial fields, are structured objects for which \(Y^1-Y^0\) may be undefined or scientifically inadequate. We introduce \emph{Topological Causal Data Analysis} (TCDA), a framework separating the observation space, causal-model class, topological representation, and causal query. Topology does not define interventions; it supplies stable, shape-sensitive summaries after causal assumptions have been specified. We distinguish outcome-level TCDA, which transforms individual potential outcomes, from distribution-level TCDA, which transforms interventional outcome laws, and characterize when outcome and distribution level contrasts agree. Building on recent outcome-level theory, we formulate identification and doubly robust representations for Banach-space-valued summaries. At the distribution level, we identify targets through the standard causal \(g\)-formula and derive stability-transfer bounds and plug-in consistency. We also place target-specific topological ignorability within the framework, clarifying when a covariate-standardized coarse effect can be identified without identifying the full interventional laws. Finally, we delimit the role of observational topology in causal discovery: it can assist diagnosis on restricted model classes but cannot by itself identify causal structure.
\end{abstract}

%\tableofcontents

%%%%%%%%%%%%%%%%%%%%%%%%%%%%%%%%%%%%%%%%%%%%%%%%%%%%%%%%

\section{Introduction}
\label{sec:introduction}

Many causal questions concern outcomes that are not naturally scalar. A treatment may alter the shape of a tumour, the organization of a brain network, the geometry of a molecular conformation, or the spatial structure of a climate field. For a real-valued outcome, the average treatment effect compares
\(\mathbb E[Y^1]\) and \(\mathbb E[Y^0]\). For shapes, images, networks, point clouds, and other structured objects, however, subtraction may be undefined or scientifically uninformative, while reducing the outcome to a single number may discard the structure of interest.

Topological data analysis provides multiscale summaries of such structure. A filtration examines an object across a range of scales, and persistent homology records when connected components, loops, cavities, and higher-dimensional features appear and disappear~\citep{carlsson2009,edelsbrunnerharer2010}. Persistence diagrams are metric-space-valued objects. Common summaries, including persistence landscapes, silhouettes, Betti curves, and persistence images, map them into function or vector spaces and often satisfy quantitative stability guarantees on suitable diagram classes~\citep{bubenik2015,chazal2014stochastic,adams2017}.

These constructions do not by themselves have a causal interpretation. A persistent feature may describe observational geometry without corresponding to an intervention, and a topological summary cannot remove confounding or determine treatment assignment. Causal meaning must come from a potential-outcome model, structural causal model, experimental design, or another set of causal assumptions. Topology can then specify the feature on which a causal contrast is evaluated.

This paper develops a mathematical framework for this combination, which we call \emph{Topological Causal Data Analysis} (TCDA). Its guiding principle is that four layers must be specified separately: the observation space, causal-model class, topological representation, and causal query. We formalize this architecture through
\(\mathfrak P_{\mathrm{TCDA}}=(S,\mathfrak M(G),T,C)\). The separation prevents three distinct questions from being conflated: which intervention is considered, which causal object is identified from the observed data, and which topological feature of that object is scientifically relevant.

% FIGURE PLACEHOLDER:
% A schematic of the four TCDA layers and the two branches:
% causal model -> potential outcomes/interventional laws;
% outcome branch: Y^a -> T_out(Y^a) -> expected contrast;
% distribution branch: L(Y^a) -> T_dist(L(Y^a)) -> law-level contrast.

Topology can enter this framework at two principal levels. In
\emph{outcome-level TCDA}, let \(Y^a\) denote the potential outcome under treatment \(a\in\{0,1\}\), and let
\(T_{\mathrm{out}}:\mathcal Y\to\mathcal B\) be a measurable representation into a separable Banach space. If
\(T_{\mathrm{out}}(Y^a)\) is Bochner integrable, the outcome-level effect is \(\operatorname{TATE}_{\mathrm{out}} = \mathbb E[T_{\mathrm{out}}(Y^1)] - \mathbb E[T_{\mathrm{out}}(Y^0)]\). It compares the expected topological summaries of individual potential outcomes. For example, if \(Y^a\) is a tumour shape and \(T_{\mathrm{out}}\) is a persistence landscape, the effect describes how treatment changes the expected persistent structure of individual shapes across filtration scales.

In \emph{distribution-level TCDA}, one first forms the interventional law \(P_Y^a=\mathcal L(Y^a)\) and then applies a representation
\(T_{\mathrm{dist}}: \mathcal P_\star(\mathcal Y)
\longrightarrow \mathcal Z_{\mathrm{dist}}\). If \(\mathcal Z_{\mathrm{dist}}=\mathcal B\) is a Banach space, the distribution-level effect is \(\Delta_{\mathrm{dist}} = T_{\mathrm{dist}}(P_Y^1) - T_{\mathrm{dist}}(P_Y^0)\). If the target is only metric, one instead considers the scalar discrepancy \(\delta_{\mathrm{dist}} = d_{\mathrm{dist}}\!\left(
T_{\mathrm{dist}}(P_Y^1), T_{\mathrm{dist}}(P_Y^0)
\right)\). A suitable density or mass sensitive representation can therefore detect a change from one persistent population cluster to two separated clusters even when the interventional laws have the same mean.

The distinction is substantive rather than notational. Define the affine law functional \(\mathcal A(P) = \int_{\mathcal Y}T_{\mathrm{out}}(y)\,dP(y)\), whenever the integral exists. Then the outcome-level effect is \(\mathcal A(P_Y^1)-\mathcal A(P_Y^0)\), whereas the distribution-level effect applies \(T_{\mathrm{dist}}\) directly to the laws. When the two maps share a common codomain, their contrasts agree for every pair of laws exactly when \(T_{\mathrm{dist}}-\mathcal A\) is constant on the relevant class. In particular, a non-affine \(T_{\mathrm{dist}}\) cannot agree universally with the outcome-level construction. Thus the order of topological transformation and population aggregation is part of the definition of the causal target.

The outcome-level branch builds directly on Kim and Lee~\citep{kimlee2026}. They define a function-valued causal effect using the power-weighted silhouette of each potential outcome's persistence diagram, identify it under standard causal assumptions, construct an efficient doubly robust estimator, establish functional weak convergence, and develop simultaneous inference and a test of no topological effect. Their silhouette estimand is a concrete instance of
\(\operatorname{TATE}_{\mathrm{out}}\), and their inferential and silhouette-specific stability results are cited rather than reproved here.

Causal effects for outcomes in general metric spaces were studied by Shin et al.~\citep{shinetal2024} through Fr\'echet means and geometric medians; their theory motivates the vectorization-free diagram comparisons considered below. At the level of interventional laws, Saki and Faghihi~\citep{sakifaghihi2026} study topological causal contrasts and target-specific topological ignorability. Their results explain how a selected covariate-standardized topological summary may be identified without identification of the complete interventional laws, while also showing why conditional topological ignorability does not generally identify marginal topology.

The purpose of the present paper is to place these constructions within a common architecture and develop the links between causal identification, topological representation, and stability. Its main contributions are as follows.

\begin{enumerate}
\item \emph{A unified TCDA framework.}
We separate the observation space, causal-model class, topological representation, and causal query, and distinguish outcome-level, marginal distribution-level, and covariate-standardized distribution-level targets.

\item \emph{A Banach-space formulation of outcome-level TCDA.}
We state the measurability, Bochner-integrability, identification, inverse-probability, augmented-weighting, and product-rate conditions needed for Banach-space-valued topological outcomes. This separates the general functional causal machinery from results specific to the power-weighted silhouette.

\item \emph{Distribution-level effects and noncommutation.}
We apply topology to interventional laws identified by standard causal arguments and characterize when outcome- and distribution-level contrasts agree. Universal agreement is governed by constancy of
\(T_{\mathrm{dist}}-\mathcal A\), while non-affinity of
\(T_{\mathrm{dist}}\) provides a general obstruction.

\item \emph{Regularity, stability, and plug-in guarantees.}
We transfer Lipschitz stability of topological representations to their causal contrasts, using the diagram metric appropriate to the chosen vectorization. Sublevel-set, Vietoris--Rips, and distance-to-measure persistence provide concrete instances. For distance-to-measure persistence, \(W_2\)-error in the estimated interventional laws directly controls the error of the plug-in topological effect.

\item \emph{Target-specific identification and discovery limits.}
With explicit attribution, we place vectorization-free effects and topological ignorability within the same framework. We also formalize the separation property required for observational topology to distinguish restricted causal-model classes, while showing that topology alone neither orients causal relations nor replaces causal assumptions.
\end{enumerate}

Topological ignorability remains a target-dependent and generally untestable assumption. For a non-injective representation it can be strictly weaker than weak conditional exchangeability only on model classes containing distinct relevant laws in the same fiber of the representation. It does not provide generic robustness to hidden confounding or identify the complete interventional laws.

Section~\ref{sec:preliminaries} fixes the notation and regularity conditions, and Section~\ref{sec:tcda-framework} introduces the general framework. Sections~\ref{sec:outcome} and
\ref{sec:distribution-level-tcda} develop its outcome- and distribution-level branches. Section~\ref{sec:stability} establishes stability and plug-in bounds. Sections~\ref{sec:topological-identifiability} and
\ref{sec:topological-ignorability} examine topology-assisted discovery and target-specific ignorability, respectively. Section~\ref{sec:minimal-examples} gives brief illustrations before the conclusion.

%%%%%%%%%%%%%%%%%%%%%%%%%%%%%%%%%%%%%%%%%%%%%%%%%%%%%%%%%

\section{Preliminaries}
\label{sec:preliminaries}

This section fixes the causal, topological, and analytic notation used throughout the paper. We state only the background needed to define the TCDA framework; standard results are cited without proof.

\subsection{Causal setup}
\label{subsec:causal-setup}

Let \((\mathcal X,\Sigma_{\mathcal X})\) and
\((\mathcal Y,\Sigma_{\mathcal Y})\) be standard Borel spaces. We observe \(O=(X,A,Y)\), where \(X\in\mathcal X\) denotes pre-treatment covariates, \(A\in\{0,1\}\) a binary treatment, and \(Y\in\mathcal Y\) an outcome. The standard-Borel assumption ensures the existence of the regular conditional distributions used below.

For \(a\in\{0,1\}\), let \(Y^a\) denote the potential outcome under the intervention \(A=a\)~\citep{rubin1974,pearl2009,imbensrubin2015}. When
\(Y^0\) and \(Y^1\) are real-valued and integrable, the classical average treatment effect is
\(\operatorname{ATE} = \mathbb E[Y^1]-\mathbb E[Y^0]\). For shape, image, graph, point-cloud, or function-valued outcomes, subtraction may be undefined or scientifically uninformative. TCDA replaces the raw outcome, or its interventional law, by a topological representation.

\begin{assumption}[Consistency]
\label{ass:consistency}
For \(a\in\{0,1\}\), \(Y=Y^a\) almost surely on \(\{A=a\}\);
equivalently, \(Y=Y^A\) almost surely.
\end{assumption}

\begin{assumption}[Conditional exchangeability]
\label{ass:exchangeability}
\((Y^0,Y^1)\perp A\mid X\).
\end{assumption}

For \(a\in\{0,1\}\), define \(e_a(x)=\mathbb P(A=a\mid X=x)\). Thus \(e_1(x)=e(x)\), where \(e\) is the propensity score, and \(e_0(x)=1-e(x)\).

\begin{assumption}[Strict positivity]
\label{ass:positivity-strict}
For \(a\in\{0,1\}\), \(e_a(X)>0\) almost surely; equivalently,
\(0<e(X)<1\) almost surely.
\end{assumption}

\begin{assumption}[Strong positivity]
\label{ass:positivity-uniform}
There exists \(\eta\in(0,1/2]\) such that
\[
\eta\leq e(X)\leq1-\eta
\qquad\text{almost surely}.
\]
\end{assumption}

Strict positivity ensures that inverse-probability weights are almost surely finite, whereas strong positivity bounds them uniformly by \(1/\eta\). Unless stated otherwise, \emph{positivity} refers to strict positivity.

Under consistency, conditional exchangeability, and strict positivity, the interventional laws are identified by the standard causal \(g\)-formula~\citep{robins1986,hernanrobins2020}. Let
\[
Q_a(x,\cdot)
:=
\mathcal L(Y\mid A=a,X=x)
\]
denote a version of the observed conditional outcome law. Then
\[
\mathcal L(Y^a)
=
\int_{\mathcal X}Q_a(x,\cdot)\,dP_X(x),
\qquad a\in\{0,1\}.
\]
Equivalently, for every bounded measurable
\(h:\mathcal Y\to\mathbb R\),
\(\mathbb E[h(Y^a)] = \mathbb E\!\left[
\mathbb E\{h(Y)\mid A=a,X\}
\right]\). All conditional laws and conditional expectations are understood up to the appropriate almost-sure equivalence.

\subsection{Topological representations}
\label{subsec:topological-representations}

Let \(\mathcal Y\) be a space of structured outcomes. A filtration assigned to \(y\in\mathcal Y\) is a nested family
\[
\mathcal F(y)=\{K_t(y)\}_{t\in I},
\qquad
K_s(y)\subseteq K_t(y)\quad\text{for }s\leq t,
\]
indexed by \(I\subseteq\mathbb R\). Point clouds may be represented by Vietoris--Rips or \v{C}ech filtrations, images by cubical filtrations, functions by sublevel- or superlevel-set filtrations, and probability laws by density, support, kernel-smoothed, or distance-to-measure
filtrations~\citep{edelsbrunnerharer2010,oudot2015,
boissonnatchazalyvinec2018,chazalmichel2021}.

Fix a coefficient field \(\Bbbk\) and a homological degree \(k\geq0\). Applying homology to \(\mathcal F(y)\) gives the persistence module
\[
\operatorname{PH}_k(\mathcal F(y))
=
\bigl\{
H_k(K_t(y);\Bbbk),\iota_{s,t}^{(k)}
\bigr\}_{s\leq t},
\]
where \(\iota_{s,t}^{(k)}\) is induced by \(K_s(y)\hookrightarrow K_t(y)\). The groups \(H_0,H_1,H_2\)
respectively describe connected components, loops or tunnels, and cavities; the Betti number \(\beta_k(K_t(y))=\dim_{\Bbbk}H_k(K_t(y);\Bbbk)\) counts independent \(k\)-dimensional homology classes.

Under the finiteness conditions stated below, the persistence module has a persistence diagram \(D_k(y) = \operatorname{Dgm}_k(\mathcal F(y))\). A finite point \(u=(b,d)\in D_k(y)\) represents a feature born at \(b\) and dying at \(d\), with persistence \(\operatorname{pers}(u)=d-b\). Unless stated otherwise, diagram metrics and vectorizations are applied to the finite part of the diagram; essential classes with death \(+\infty\) are removed through an appropriate convention or treated separately.

Persistence diagrams are metric rather than linear objects. Let \((\mathsf D_k,d_{\mathsf D})\) denote a specified metric space of degree-\(k\) diagrams, equipped, for example, with the bottleneck distance \(d_B\) or a \(p\)-Wasserstein distance \(W_p\) \citep{cohensteineredelsbrunnerharer2007,mileyko2011,turner2014}. The admissible diagram class is always chosen so that the relevant metric is finite.

A degree-\(k\) topological summary is a map
\(\Phi: (\mathsf D_k,d_{\mathsf D})
\longrightarrow \mathcal Z\), where \(\mathcal Z\) is a metric space. When expectations and linear contrasts are required, \(\mathcal Z=\mathcal B\) is a Banach space. Typical single-degree choices include persistence landscapes, silhouettes, persistence images, Betti curves, and kernel-based representations~\citep{bubenik2015,bubenikdlotko2017,chazal2014stochastic,adams2017,kwitt2015,carriere2017,kusano2018,le2018}. Euler characteristic curves and other summaries combining homological degrees are obtained from a finite tuple
\((D_j)_{j\in J}\) rather than from a single \(D_k\); we suppress this straightforward notational variant. The diagram metric, target norm or metric, and any restrictions on the admissible diagrams form part of the specification of \(\Phi\).

For an outcome-level pipeline, write
\(T_{\mathrm{out}} = \Phi\circ\operatorname{Dgm}_k\circ\mathcal F\), and
\(T_{\mathrm{out}}: \mathcal Y\longrightarrow\mathcal Z_{\mathrm{out}}\). If \(\Phi\) is the identity, the representation is diagram-valued and must be analysed using metric-space methods rather than linear expectations.

A distribution-level representation instead acts on a specified metric class \((\mathcal P_\star(\mathcal Y),d_{\mathcal P})\) of Borel probability laws:
\[
T_{\mathrm{dist}}
=
\Phi\circ\operatorname{Dgm}_k\circ\mathcal F,
\qquad
T_{\mathrm{dist}}:
\mathcal P_\star(\mathcal Y)
\longrightarrow
\mathcal Z_{\mathrm{dist}}.
\]
Here the filtration is applied to a probability law rather than to an individual outcome. The notation is schematic: the filtration, homological degree, summary map, and target space need not be the same in the outcome- and distribution-level pipelines.

Continuity or Lipschitz stability of \(\Phi\) must be stated with respect to the same diagram metric used to control the persistence diagram. For example, continuity in \(W_p\) does not follow from bottleneck stability without additional assumptions. The precise metric compatibility used for the standard vectorizations is recorded in Section~\ref{sec:stability}.

The filtration and summary map are part of the scientific modelling choice. A Vietoris--Rips filtration emphasizes metric proximity, a cubical filtration emphasizes spatial intensity structure, and a distance-to-measure filtration emphasizes the geometry of local probability mass. TCDA does not select this representation automatically; the choice must be justified by the outcome modality and the causal question.

\subsection{Regularity of the persistence pipeline}
\label{subsec:pipeline-regularity}

Three regularity requirements are needed before a topological
representation can be used as a causal outcome: the persistence diagram must exist, the representation must be measurable, and any Banach-space-valued random element whose expectation is taken must be Bochner integrable.

\paragraph{\(q\)-tameness.}
A one-parameter persistence module
\(V=\{V_t,\varphi_{s,t}\}_{s\leq t}\) is \(q\)-tame if
\(\operatorname{rank}(\varphi_{s,t})<\infty\) for every \(s<t\). A \(q\)-tame module admits a well-defined persistence diagram~\citep{oudot2015,chazaldesilvaglisseoudot2016}. No stronger interval-decomposition statement is needed below.

\begin{assumption}[\(q\)-tameness]
\label{ass:q-tameness}
For each homological degree used in the analysis:

\begin{enumerate}
\item \(\operatorname{PH}_k(\mathcal F(y))\) is \(q\)-tame for every relevant outcome \(y\);

\item \(\operatorname{PH}_k(\mathcal F(P))\) is \(q\)-tame for every law \(P\) in the domain of \(T_{\mathrm{dist}}\).
\end{enumerate}
\end{assumption}

This condition is automatic for filtrations of finite simplicial or cubical complexes and for the standard filtrations of finite point clouds. Infinite, compact-set, and population-level filtrations require their own \(q\)-tameness justification. In particular, finite second
moment and Lipschitz regularity of a distance-to-measure function do not, by themselves, imply \(q\)-tameness of every induced sublevel-set persistence module.

\paragraph{Diagram spaces.}
When \(W_p\) is used with \(1\leq p<\infty\), we work in the space \(\mathsf D_{k,p}\) of diagrams with finite degree-\(p\) total persistence. The space \((\mathsf D_{k,p},W_p)\) is complete and separable~\citep{mileyko2011}. When the bottleneck distance is used, the domain is a specified metric subspace on which the diagram map and the selected vectorization are defined.

\paragraph{Measurability and Bochner integration.}
The following standard conditions are imposed whenever the
corresponding objects are used:

\begin{enumerate}
\item the diagram maps \( y\longmapsto\operatorname{Dgm}_k(\mathcal F(y)),
\ P\longmapsto\operatorname{Dgm}_k(\mathcal F(P))\) are Borel measurable in their stated metrics;

\item the vectorization \(\Phi\) is Borel measurable;

\item whenever expectations are taken, the target \(\mathcal B\) is a separable Banach space and
\[
\mathbb E\|T_{\mathrm{out}}(Y^a)\|_{\mathcal B}<\infty.
\]
\end{enumerate}

Continuity of the diagram map and of \(\Phi\) is a sufficient condition for the first two requirements. Their composition is then Borel measurable, and hence strongly measurable when the target Banach space is separable. These are standard consequences of the Pettis measurability theorem~\citep{diesteluhl1977}.

A strongly measurable Banach-space-valued random element \(U\) is Bochner integrable exactly when \(\mathbb E\|U\|_{\mathcal B}<\infty\). Its expectation is linear and
satisfies
\(\|\mathbb E[U]\|_{\mathcal B} \leq \mathbb E\|U\|_{\mathcal B}\). Bochner conditional expectations satisfy the corresponding tower and conditional Jensen properties~\citep{diesteluhl1977}.

A useful standard sufficient condition is the following. If
\(T_{\mathrm{out}}\) is \(L\)-Lipschitz and, for some
\(y_0\in\mathcal Y\), \(\mathbb E[d_{\mathcal Y}(Y^a,y_0)]<\infty\), then
\[
\mathbb E\|T_{\mathrm{out}}(Y^a)\|_{\mathcal B}
\leq
\|T_{\mathrm{out}}(y_0)\|_{\mathcal B}
+
L\,\mathbb E[d_{\mathcal Y}(Y^a,y_0)]
<
\infty.
\]
Uniform bounds on the number and persistence of off-diagonal diagram points give alternative sufficient conditions for standard summaries. When such bounds are required later, they are stated for the specific filtration and vectorization being used.

Finally, bounded linear operators commute with Bochner integration. If \(\mathcal T\) is a compact metric space, then
\(C(\mathcal T)\), equipped with the supremum norm, is a separable Banach space, and each point-evaluation map is bounded and linear. Hence
\[
\mathbb E[U](t)
=
\mathbb E[U(t)],
\qquad
t\in\mathcal T,
\]
for every Bochner-integrable \(C(\mathcal T)\)-valued random element \(U\). By contrast, point evaluation is not well defined on the equivalence classes forming \(L^p(\mathcal T)\). Pointwise identities in \(L^p\) therefore require jointly measurable representatives and the appropriate Fubini--Tonelli assumptions.

%%%%%%%%%%%%%%%%%%%%%%%%%%%%%%%%%%%%%%%%%%%%%%%%%%%%%%%%

\section{The TCDA framework}
\label{sec:tcda-framework}

The outcome-level branch of TCDA builds directly on Kim and
Lee~\citep{kimlee2026}, who define and estimate a function-valued causal effect based on power-weighted persistence silhouettes. At the level of interventional laws, Saki and Faghihi \citep{sakifaghihi2026} study topological causal contrasts and target-specific topological ignorability. The purpose of the present section is to place these and related constructions within a common mathematical architecture.

The central principle is that four layers must be specified
separately: the observation space, the causal-model class, the
topological representation, and the causal query. A topological construction does not define an intervention; a causal graph does not determine a filtration; and a persistence diagram does not by itself identify a causal effect. TCDA combines these layers without conflating their mathematical roles.

\subsection{Formal definition}
\label{subsec:formal-definition}

Let \(V_{\mathrm{obs}}\) be a finite set of observed variables. For each \(v\in V_{\mathrm{obs}}\), let \((S_v,\Sigma_v)\) be a standard Borel space, and define
\[
(S,\Sigma)
=
\left(
\prod_{v\in V_{\mathrm{obs}}}S_v,
\bigotimes_{v\in V_{\mathrm{obs}}}\Sigma_v
\right).
\]
An observation is a random element \(O=(O_v)_{v\in V_{\mathrm{obs}}}\in S\). In the binary-treatment
setting, \(O=(X,A,Y)\), where \(X\) denotes pre-treatment covariates, \(A\in\{0,1\}\) treatment, and \(Y\in\mathcal Y\) the outcome.

Let \(G\) be a causal graph, possibly containing latent variables, and let \(\mathfrak M(G)\) be a specified class of causal models compatible with \(G\). Each \(M\in\mathfrak M(G)\) determines an observational law \(P_M^{\mathrm{obs}}\) and the potential-outcome or interventional objects required by the analysis. In particular, whenever \(\operatorname{do}(A=a)\) is defined, write \(P_{Y,M}^a = \mathcal L_M(Y^a)\) for the corresponding interventional outcome law.

An outcome-level topological representation is a measurable map \(T_{\mathrm{out}}: \mathcal Y\longrightarrow\mathcal Z_{\mathrm{out}}\), where \(\mathcal Z_{\mathrm{out}}\) is a specified metric space. When expectations and linear contrasts are used, \(\mathcal Z_{\mathrm{out}}=\mathcal B\) is taken to be a separable Banach space.

Let \(\mathcal P^\star(\mathcal Y)\subseteq\mathcal P(\mathcal Y)\) be a specified metric class of Borel probability laws containing \(P_{Y,M}^0\) and \(P_{Y,M}^1\) for every model under consideration. A distribution-level topological representation is a measurable map
\(T_{\mathrm{dist}}: \mathcal P^\star(\mathcal Y)
\longrightarrow \mathcal Z_{\mathrm{dist}}\), where \(\mathcal Z_{\mathrm{dist}}\) is a specified metric space. Whenever \(\mathcal Z_{\mathrm{dist}}=\mathcal B\) is a Banach space and vector contrasts are used, it is equipped with the norm-induced metric \(d_{\mathrm{dist}}(u,v)=\|u-v\|_{\mathcal B}\).

In persistent-homology applications, either representation typically has the form \(T = \Phi\circ\operatorname{Dgm}_k\circ\mathcal F\), where \(\mathcal F\) is applied either to an individual outcome or to a probability law. The filtration, diagram metric, vectorization, and target space are all part of the specification of \(T\).

\begin{definition}[TCDA problem]
\label{def:tcda-problem}
A \emph{topological causal data analysis problem} is a tuple
\[
\mathfrak P_{\mathrm{TCDA}}
=
(S,\mathfrak M(G),T,C),
\]
where:

\begin{enumerate}
\item \(S\) is the measurable observation space, equipped with any additional structure required by the chosen topological representation;

\item \(\mathfrak M(G)\) is the class of causal models under
consideration;

\item \(T\) is a specified measurable topological representation of outcomes, probability laws, or another object determined by the causal model;

\item \(C\) is the causal query, represented as a map
\(C:\mathfrak M(G)\longrightarrow\mathcal E\) into a specified target space \(\mathcal E\).
\end{enumerate}
\end{definition}

Definition~\ref{def:tcda-problem} separates the definition of the causal target from its identification. The object \(C(M)\) is defined at the level of the causal model; whether it is determined by \(P_M^{\mathrm{obs}}\) is a separate question addressed in later sections.

\subsection{Topological potential outcomes}
\label{subsec:topological-potential-outcomes}

Let \(T_{\mathrm{out}}:\mathcal Y\to\mathcal B\) be measurable, where \(\mathcal B\) is a separable Banach space. Whenever expectations are used, assume
\(\mathbb E_M \bigl[ \|T_{\mathrm{out}}(Y^a)\|_{\mathcal B}
\bigr] <\infty,\ a=0,1\).

\begin{definition}[Topological potential outcome]
\label{def:topological-potential-outcome}
The \emph{topological potential outcome} under treatment level
\(a\in\{0,1\}\) is \(Z^a=T_{\mathrm{out}}(Y^a)\). The observed topological outcome is \(Z=T_{\mathrm{out}}(Y)\).
\end{definition}

These transformed potential outcomes inherit the standard causal assumptions in the expected direction. If \(Y=Y^A\) almost surely, then \(Z=Z^A\) almost surely. If \((Y^0,Y^1)\perp A\mid X\), then \((Z^0,Z^1)\perp A\mid X\), because conditional independence is preserved under measurable transformations. Positivity is unchanged,
since the transformation affects neither \(A\) nor \(X\).

The converses need not hold when \(T_{\mathrm{out}}\) is
non-injective. Exchangeability of transformed potential outcomes may therefore contain less information than exchangeability of the original outcomes. This elementary observation should not be confused with the distribution-level topological-ignorability condition studied in Section~\ref{sec:topological-ignorability}.

More generally, if \(s:\mathcal B\to\mathbb R\) is measurable, then \(s(Z^a)\) may be treated as a scalar potential outcome whenever it is integrable. Examples include norms, integrals, and measurable scalar statistics of persistence-based representations.

\subsection{Basic topological causal estimands}
\label{subsec:topological-causal-estimands}

The framework gives rise to two principal causal targets.

\paragraph{Outcome-level effects.}
Assume that \(Z^0\) and \(Z^1\) are Bochner integrable.

\begin{definition}[Outcome-level topological average treatment effect]
\label{def:outcome-level-tate}
The \emph{outcome-level topological average treatment effect} is
\[
\operatorname{TATE}_{\mathrm{out}}
=
\mathbb E[Z^1]-\mathbb E[Z^0]
=
\mathbb E[T_{\mathrm{out}}(Y^1)]
-
\mathbb E[T_{\mathrm{out}}(Y^0)]
\in\mathcal B.
\]
\end{definition}

This estimand measures the causal effect on the information retained by the chosen representation \(T_{\mathrm{out}}\). It is therefore representation-dependent: different filtrations and vectorizations generally define different causal targets.

When \(\mathcal B=C(\mathbb T)\) for compact \(\mathbb T\), point evaluation is bounded and linear, so
\[
\operatorname{TATE}_{\mathrm{out}}(t)
=
\mathbb E[Z^1(t)]-\mathbb E[Z^0(t)],
\qquad t\in\mathbb T.
\]
For \(\mathcal B=L^p(I)\), the corresponding identity holds for almost every  \(t\), after choosing jointly measurable representatives and under the Fubini conditions stated in
Subsection~\ref{subsec:pipeline-regularity}.

Taking \(T_{\mathrm{out}}\) to be the power-weighted persistence silhouette recovers the effect curve \(\psi_k\) of Kim and Lee~\citep{kimlee2026}. Their identification, efficient estimation, functional weak convergence, testing, and silhouette-stability results remain their results and are invoked with attribution in Section~\ref{sec:outcome}.

A scalar outcome-level effect can be obtained from a measurable functional \(s:\mathcal B\to\mathbb R\):

\begin{definition}[Scalar topological average treatment effect]
\label{def:scalar-tate}
If \(\mathbb E|s(Z^a)|<\infty\) for \(a=0,1\), define
\[
\operatorname{TATE}_{s}
=
\mathbb E[s(Z^1)]-\mathbb E[s(Z^0)].
\]
\end{definition}

For nonlinear \(s\), this generally differs from
\(s(\operatorname{TATE}_{\mathrm{out}})\). If
\(s\in\mathcal B^\ast\) is bounded and linear, however, then
\(\operatorname{TATE}_{s} = s(\operatorname{TATE}_{\mathrm{out}})\). Thus the distinction between scalarizing before and after averaging arises from nonlinearity.

\paragraph{Distribution-level effects.}
Let
\(T_{\mathrm{dist}}:\mathcal P^\star(\mathcal Y)\to\mathcal Z_{\mathrm{dist}}\) be defined at the two interventional laws. If \(\mathcal Z_{\mathrm{dist}}=\mathcal B\) is a Banach space, the natural contrast is
\(T_{\mathrm{dist}}(P_{Y,M}^1) - T_{\mathrm{dist}}(P_{Y,M}^0)\). If the target is only a metric space
\((\mathcal Z_{\mathrm{dist}},d_{\mathrm{dist}})\), one may instead use the scalar discrepancy
\[
d_{\mathrm{dist}}\!\left(
T_{\mathrm{dist}}(P_{Y,M}^1),
T_{\mathrm{dist}}(P_{Y,M}^0)
\right),
\]
which records magnitude but not direction. These effects are defined formally in Section~\ref{sec:distribution-level-tcda}.

The outcome-level quantity
\(\mathbb E_M[T_{\mathrm{out}}(Y^a)]\) and the distribution-level quantity \(T_{\mathrm{dist}}(P_{Y,M}^a)\) answer different causal questions and may belong to different target spaces. Even when directly comparable, they need not agree. Their relationship is studied in Subsection~\ref{subsec:noncommutation}.

%%%%%%%%%%%%%%%%%%%%%%%%%%%%%%%%%%%%%%%%%%%%%%%%%%%%%%%%
%%%%%%%%%%%%%%%%%%%%%%%%%%%%%%%%%%%%%%%%%%%%%%%%%%%%%%%%

\section{Outcome-level topological causal effects}
\label{sec:outcome}

Let \(T_{\mathrm{out}}:\mathcal Y\to\mathcal B\) be a measurable outcome-level representation into a separable Banach space, and write \(Z=T_{\mathrm{out}}(Y), \ Z^a=T_{\mathrm{out}}(Y^a),\ a\in\{0,1\}\). We assume that \(Z^0\) and \(Z^1\) are Bochner integrable. Because \(T_{\mathrm{out}}\) is deterministic and measurable, consistency and conditional exchangeability for \(Y^a\) imply the corresponding properties for \(Z^a\). Thus the topological representation changes the outcome being compared, but not the logic of causal identification.

\subsection{Identification and observed-data representations}
\label{subsec:identification-functional-tate}

Define \(m_a(x)=\mathbb E[Z\mid A=a,X=x], \ e_a(x)=\mathbb P(A=a\mid X=x)\), where the first conditional expectation is Bochner valued. Thus \(e_1=e\) and \(e_0=1-e\), where \(e(x)=\mathbb P(A=1\mid X=x)\).

\begin{theorem}[Identification of the outcome-level topological ATE]
\label{thm:identification-outcome-tate}
Assume consistency, conditional exchangeability, and strict positivity. If \(\mathbb E\|Z^a\|_{\mathcal B}<\infty\) for \(a=0,1\), then
\[
\mathbb E[Z^a]
=
\mathbb E[m_a(X)]
=
\mathbb E\left[
\frac{\mathbf 1\{A=a\}}{e_a(X)}Z
\right].
\]
Consequently,
\[
\operatorname{TATE}_{\mathrm{out}}
=
\mathbb E[m_1(X)-m_0(X)]
=
\mathbb E\left[
\left\{
\frac{A}{e(X)}
-
\frac{1-A}{1-e(X)}
\right\}Z
\right].
\]
All expectations above are Bochner expectations.
\end{theorem}

\begin{proof}
Fix \(a\in\{0,1\}\). Conditional exchangeability of \(Y^a\) and measurability of \(T_{\mathrm{out}}\) imply \(Z^a\perp A\mid X\), while consistency gives \(Z=Z^a\) on \(\{A=a\}\). Hence, under strict positivity,
\[
\mathbb E[Z^a\mid X]
=
\mathbb E[Z^a\mid A=a,X]
=
\mathbb E[Z\mid A=a,X]
=
m_a(X)
\]
almost surely. Taking expectations proves the standardization identity.

For the inverse-probability identity, consistency and conditional exchangeability give
\[
\begin{aligned}
\mathbb E\left[
\left.
\frac{\mathbf 1\{A=a\}}{e_a(X)}Z
\right|X
\right]
&=
\frac{1}{e_a(X)}
\mathbb E\left[
\mathbf 1\{A=a\}Z^a\mid X
\right] \\
&=
\mathbb E[Z^a\mid X].
\end{aligned}
\]
Moreover,
\[
\mathbb E\left[
\frac{\mathbf 1\{A=a\}}{e_a(X)}
\|Z\|_{\mathcal B}
\right]
=
\mathbb E\|Z^a\|_{\mathcal B}<\infty,
\]
so the weighted random element is Bochner integrable. Taking
expectations and subtracting the two treatment-specific identities completes the proof.
\end{proof}

These are the usual standardization and inverse-probability identities, applied to a Banach-space-valued outcome. In particular, strict positivity is sufficient for the population identities; uniformly bounded inverse weights are needed only when imposed by subsequent estimation or asymptotic arguments. For the power-weighted silhouette,
point evaluation recovers the corresponding identification formulas of Kim and Lee~\citep{kimlee2026}.

\paragraph{Conditional effects.}
Let \((\mathcal V,\Sigma_{\mathcal V})\) be a standard Borel space, let
\(\nu:\mathcal X\to\mathcal V\) be measurable, and define
\(V=\nu(X)\). Thus \(V\) is a prespecified measurable summary of the covariates.

\begin{definition}[Conditional topological average treatment effect]
\label{def:tcate}
For \(P_V\)-almost every \(v\), the \emph{conditional topological average treatment effect} is
\[
\operatorname{TCATE}_{\mathrm{out}}(v)
=
\mathbb E[Z^1-Z^0\mid V=v].
\]
\end{definition}

Under the conditions of Theorem~\ref{thm:identification-outcome-tate},
\[
\operatorname{TCATE}_{\mathrm{out}}(v)
=
\mathbb E[m_1(X)-m_0(X)\mid V=v],
\qquad P_V\text{-almost everywhere},
\]
and
\(\operatorname{TATE}_{\mathrm{out}} = \mathbb E\!\left[
\operatorname{TCATE}_{\mathrm{out}}(V)\right]\).

The regressions remain conditional on the full adjustment set \(X\); \(V\) indexes effect heterogeneity and need not itself control all confounding. In general,
\[
\operatorname{TCATE}_{\mathrm{out}}(v)
\neq
\mathbb E[Z\mid A=1,V=v]
-
\mathbb E[Z\mid A=0,V=v].
\]

\subsection{Augmented representation and estimation}
\label{subsec:aipw-outcome}

Let \(\widetilde m_a:\mathcal X\to\mathcal B\) and
\(\widetilde e:\mathcal X\to(0,1)\) be candidate nuisance functions, with \(\widetilde e_1=\widetilde e, \ \ \widetilde e_0=1-\widetilde e\). Define
\[
\Psi_a(\widetilde m_a,\widetilde e_a)
=
\mathbb E\left[
\widetilde m_a(X)
+
\frac{\mathbf 1\{A=a\}}{\widetilde e_a(X)}
\{Z-\widetilde m_a(X)\}
\right].
\]

\begin{proposition}[Augmented remainder and product-rate bound]
\label{prop:product-rate-bias}
Assume the conditions of Theorem~\ref{thm:identification-outcome-tate}. Suppose that \(\mathbb E\|\widetilde m_a(X)\|_{\mathcal B}<\infty\) and, for some
\(\varepsilon>0\), \(\widetilde e_a(X)\geq\varepsilon\) almost surely. Then
\[
\Psi_a(\widetilde m_a,\widetilde e_a)-\mathbb E[Z^a]
=
\mathbb E\left[
\frac{\widetilde e_a(X)-e_a(X)}
     {\widetilde e_a(X)}
\{\widetilde m_a(X)-m_a(X)\}
\right].
\]
Consequently,
\(\Psi_a(\widetilde m_a,\widetilde e_a)=\mathbb E[Z^a]\) if either
\(\widetilde m_a=m_a\) or \(\widetilde e_a=e_a\) almost surely.

If the two nuisance errors are square integrable, then
\[
\begin{aligned}
\left\|
\Psi_a(\widetilde m_a,\widetilde e_a)-\mathbb E[Z^a]
\right\|_{\mathcal B}
\leq
\frac{1}{\varepsilon}
\|\widetilde e_a-e_a\|_{L^2(P_X)}
\left(
\mathbb E
\|\widetilde m_a(X)-m_a(X)\|_{\mathcal B}^{2}
\right)^{1/2}.
\end{aligned}
\]
The corresponding bound for the treatment contrast is obtained by summing the two arm-specific bounds.
\end{proposition}

\begin{proof}
Define \(R_a = \widetilde m_a(X) + \frac{\mathbf 1\{A=a\}}{\widetilde e_a(X)} \{Z-\widetilde m_a(X)\}\). By consistency,
\[
\mathbb E\|Z\|_{\mathcal B}
=
\mathbb E\!\left[
\mathbf 1\{A=0\}\|Z^0\|_{\mathcal B}
+
\mathbf 1\{A=1\}\|Z^1\|_{\mathcal B}
\right] \\
\leq
\mathbb E\|Z^0\|_{\mathcal B}
+
\mathbb E\|Z^1\|_{\mathcal B}
<\infty.
\]
The augmented random element \(R_a\) is therefore Bochner integrable because
\[
\|R_a\|_{\mathcal B}
\leq
\|\widetilde m_a(X)\|_{\mathcal B}
+
\varepsilon^{-1}
\left\{
\|Z\|_{\mathcal B}
+
\|\widetilde m_a(X)\|_{\mathcal B}
\right\}.
\]
Conditioning on \(X\) and using
\(\mathbb E[\mathbf 1\{A=a\}Z\mid X]=e_a(X)m_a(X)\) gives
\[
\mathbb E[R_a\mid X]
=
\widetilde m_a(X)
+
\frac{e_a(X)}{\widetilde e_a(X)}
\{m_a(X)-\widetilde m_a(X)\}.
\]
Subtracting \(m_a(X)\), taking expectations, and using
\(\mathbb E[m_a(X)]=\mathbb E[Z^a]\) proves the exact remainder identity. Double robustness follows by setting either nuisance error to zero.

Finally, the Bochner--Jensen and Cauchy--Schwarz inequalities give
\[
\|\Psi_a-\mathbb E[Z^a]\|_{\mathcal B}
\leq
\frac{1}{\varepsilon}
\mathbb E\left[
|\widetilde e_a-e_a|
\|\widetilde m_a-m_a\|_{\mathcal B}
\right] 
\leq
\frac{1}{\varepsilon}
\|\widetilde e_a-e_a\|_{L^2(P_X)}
\left(
\mathbb E\|\widetilde m_a-m_a\|_{\mathcal B}^{2}
\right)^{1/2}.
\]
\end{proof}

This is the standard doubly robust remainder, written for a
Banach-space-valued outcome. The product bound controls the population bias; it does not by itself establish a central limit theorem in an arbitrary Banach space.

\paragraph{Cross-fitted estimator.}
\label{subsec:sample-estimator-outcome}

Let \(O_i=(X_i,A_i,Y_i)\), \(i=1,\ldots,n\), be independent and identically distributed observations, and define
\(Z_i=T_{\mathrm{out}}(Y_i)\). Let \(I_1,\ldots,I_K\) be a fixed partition of the observations, and let \(k(i)\) denote the fold containing observation \(i\). For each fold \(k\), fit \(\widehat m_a^{(-k)}\) and a propensity estimator
\(\widehat e^{(-k)}\) without using the observations in \(I_k\), and set
\[
\widehat e_1^{(-k)}=\widehat e^{(-k)},
\qquad
\widehat e_0^{(-k)}=1-\widehat e^{(-k)}.
\]
Assume that, for some fixed \(\varepsilon>0\), the fitted propensity scores are restricted so that
\(\varepsilon \leq \widehat e^{(-k)}(x)
\leq 1-\varepsilon\) for every fold \(k\) and \(P_X\)-almost every \(x\), with probability tending to one.

The cross-fitted estimator is
\[
\widehat\mu_a^{\mathrm{cf}}
=
\frac1n\sum_{i=1}^n
\left[
\widehat m_a^{(-k(i))}(X_i)
+
\frac{\mathbf 1\{A_i=a\}}
     {\widehat e_a^{(-k(i))}(X_i)}
\left\{
Z_i-\widehat m_a^{(-k(i))}(X_i)
\right\}
\right],
\]
with
\(\widehat{\operatorname{TATE}}_{\mathrm{out}}^{\mathrm{cf}}
= \widehat\mu_1^{\mathrm{cf}} - \widehat\mu_0^{\mathrm{cf}}\).

Conditional on the nuisance-training folds, Proposition~\ref{prop:product-rate-bias} controls the arm-specific bias by the product of the two nuisance errors. In particular, if, for each \(a\),
\(\|\widehat e_a-e_a\|_{L^2(P_X)} = o_p(n^{-1/4})\) and
\(\left( \mathbb E\| \widehat m_a(X)-m_a(X)
\|_{\mathcal B}^2 \right)^{1/2} = o_p(n^{-1/4})\), then the corresponding bias remainder is \(o_p(n^{-1/2})\). Root-\(n\) inference still requires representation-specific moment, tightness, and weak-convergence conditions~\citep{chernozhukov2018,kennedy2022}.

\subsection{Available functional inference for silhouettes}
\label{subsec:inference}

A general central limit theorem does not follow merely from
separability of \(\mathcal B\). Functional inference is currently available for the power-weighted silhouette through the theory of Kim and Lee~\citep{kimlee2026}.

For a finite diagram \(D\), a point \(p=(b_p,d_p)\in D\), and \(r>0\), define
\[
\Lambda_p(t)
=
\max\{0,\min\{t-b_p,d_p-t\}\},
\]
and \(S_r(D) = \sum_{p\in D}(d_p-b_p)^r\). Whenever \(S_r(D)>0\), the power-weighted silhouette is
\[
\phi(t;D,r)
=
\frac{
\sum_{p\in D}(d_p-b_p)^r\Lambda_p(t)
}{
S_r(D)
}.
\]

Let \(\mathbb T\subset\mathbb R\) be compact, and write
\(D_k = \operatorname{Dgm}_k(\mathcal F(Y)),
\ \ D_k^a = \operatorname{Dgm}_k(\mathcal F(Y^a))\). Assume
\(S_r(D_k^a)>0\) almost surely, \(a=0,1\), as required for the normalized silhouette used by Kim and Lee. If empty diagrams are to be permitted, a separate convention and verification of the imported stability and inference results are required.

Define \(Z_k(t)=\phi(t;D_k,r),\ \ Z_k^a(t)=\phi(t;D_k^a,r),
\qquad t\in\mathbb T\), and \(m_{a,k}(t,x) = \mathbb E[Z_k(t)\mid A=a,X=x]\). Writing \(\eta_k=(e,m_{0,k},m_{1,k})\), their uncentered efficient-influence-function signal is
\[
\begin{aligned}
\varphi_k(t,O;\eta_k)
={}&
m_{1,k}(t,X)-m_{0,k}(t,X) \\
&+
\frac{A}{e(X)}
\{Z_k(t)-m_{1,k}(t,X)\}
-
\frac{1-A}{1-e(X)}
\{Z_k(t)-m_{0,k}(t,X)\}.
\end{aligned}
\]
Its expectation equals the silhouette effect curve
\(\psi_k(t) = \mathbb E[Z_k^1(t)-Z_k^0(t)]\).

Under their Assumptions~(A1)--(A4) and either (A5) or (A5\(^{\prime}\)), Kim and Lee~\citep[Theorem~5.2]{kimlee2026} establish weak convergence of their sample-split augmented estimator in \(\ell^\infty(\mathbb T)\). In particular, the estimator is pointwise asymptotically efficient and supports simultaneous inference when a valid approximation to the supremum distribution is available. Under Assumptions~(A1)--(A6), and under the additional condition in their Corollary~5.4 that a Gaussian- or Rademacher-multiplier bootstrap consistently approximates the distribution of the supremum of the limiting Gaussian process, Kim and Lee obtain a multiplier-bootstrap test based on
\(\sqrt n\|\widehat\psi_k\|_\infty\).

Their diagram-level null,
\(W_1(D_k^1,D_k^0)=0\) almost surely, implies \(\psi_k\equiv0\), but the converse need not hold: silhouettes are not injective, and heterogeneous effects may cancel in expectation. The procedure is therefore sensitive to the expected silhouette contrast, not to every possible change in the potential diagrams. Extensions to landscapes, persistence images, Betti curves, or general \(K\)-fold cross-fitting require separate asymptotic arguments and are not claimed here.

\subsection{Vectorization-free effects}
\label{subsec:frechet-effects}

Vectorization is not necessary when the causal contrast is defined directly through the metric geometry of persistence diagrams. Let \(D_k^a=\operatorname{Dgm}_k(\mathcal F(Y^a))\) and \(\nu^a=\mathcal L(D_k^a)\). On a specified diagram model
\((\mathsf D_{k,p}^{\star},W_p)\), assume that \(\nu^a\) has finite second \(W_p\)-moment and that the following minimizers exist and are unique. Define the Fr\'echet mean and geometric median by
\[
\bar D^a
\in
\arg\min_{D\in\mathsf D_{k,p}^{\star}}
\int W_p(D,D')^2\,d\nu^a(D')
\quad
\text{and}
\quad 
\widetilde D^a
\in
\arg\min_{D\in\mathsf D_{k,p}^{\star}}
\int W_p(D,D')\,d\nu^a(D'),
\]
respectively. The corresponding absolute effects are
\[
\operatorname{AATE}_{\mathrm{top}}
=
W_p(\bar D^1,\bar D^0),
\qquad
\operatorname{AMTE}_{\mathrm{top}}
=
W_p(\widetilde D^1,\widetilde D^0).
\]
These are persistence-diagram specializations of the metric-space causal estimands introduced by Shin et al.~\citep{shinetal2024}.

Under consistency, conditional exchangeability, and positivity, the diagram laws are identified by
\(\nu^a(B) = \mathbb E\left[ \mathbb P\{D_k\in B\mid A=a,X\}
\right]\). Hence any unique Fr\'echet mean or median, and the resulting absolute effect, is identified. Estimation may use the stratification-weighted Fr\'echet procedures of Shin et al.~\citep{shinetal2024} when their finite-stratification, moment, uniqueness, and proper-space conditions hold.

The properness restriction is substantive: the unrestricted
Wasserstein persistence-diagram space is generally not proper, so strong consistency cannot be inferred from Polishness alone. Likewise, bootstrap validity requires additional regularity and does not follow from consistency. These metric-space estimands are therefore retained
as an available extension, not as new estimation or inferential theory of the present paper.

%%%%%%%%%%%%%%%%%%%%%%%%%%%%%%%%%%%%%%%%%%%%%%%%%%%%%%
%%%%%%%%%%%%%%%%%%%%%%%%%%%%%%%%%%%%%%%%%%%%%%%%%%%%%%

\section{Distribution-level topological causal effects}
\label{sec:distribution-level-tcda}

Outcome-level TCDA applies topology to each potential outcome and then averages: \(\operatorname{TATE}_{\mathrm{out}} = \mathbb E[T_{\mathrm{out}}(Y^1)] - \mathbb E[T_{\mathrm{out}}(Y^0)]\). Distribution-level TCDA reverses this order. It first forms the interventional law \(P_Y^a=\mathcal L(Y^a)\) and then applies a topological representation \(T_{\mathrm{dist}}\) to that law. The resulting estimands describe treatment-induced changes in population-level geometry, such as clustering, connectivity, or persistent holes.

The two constructions generally answer different questions. The outcome-level effect depends on an interventional law through the average of an individual-outcome representation, whereas the distribution-level effect applies a generally nonlinear transformation to the law itself. Law-level topological contrasts and covariate-standardized topological effects are also studied by Faghihi and Saki~\citep{sakifaghihi2026}. The present section places these targets within the TCDA framework, distinguishes them from outcome-level effects, identifies them under standard causal assumptions, and characterizes when the two levels can agree.

\subsection{Interventional laws and distribution-level estimands}
\label{subsec:interventional-laws}

Let \((\mathcal Y,\Sigma_{\mathcal Y})\) be a standard Borel outcome space, equipped with any additional metric or geometric structure needed for the chosen filtration. Let
\((\mathcal P^\star(\mathcal Y),d_{\mathcal P})\) be a specified class of probability measures containing the interventional laws \(P_Y^a=\mathcal L(Y^a),\  a\in\{0,1\}\). A distribution-level topological representation is a measurable map \(T_{\mathrm{dist}}: \mathcal P^\star(\mathcal Y) \longrightarrow \mathcal Z_{\mathrm{dist}}\), where \((\mathcal Z_{\mathrm{dist}},d_{\mathrm{dist}})\) is a metric
space of topological summaries. Typical examples apply persistent homology to the support, density level sets, a smoothed population function, or the distance-to-measure function of \(P\). Schematically,
\(T_{\mathrm{dist}}(P) = \Phi\!\left(
\operatorname{Dgm}_k( \mathcal F_{\mathrm{dist}}(P))
\right)\). The filtration and the metric on \(\mathcal Y\) are part of the estimand: a probability law has no unique intrinsic topological representation.

\begin{definition}[Interventional topological representation]
\label{def:interventional-topological-representation}
The distribution-level interventional topological representation under treatment \(a\) is
\[
\Theta_{\mathrm{dist}}^a
=
T_{\mathrm{dist}}(P_Y^a)
=
T_{\mathrm{dist}}\bigl(\mathcal L(Y^a)\bigr).
\]
\end{definition}

For a fixed causal model, \(\Theta_{\mathrm{dist}}^a\) is a deterministic functional of the interventional law, not a unit-level random outcome.

\begin{definition}[Distribution-level topological treatment effect]
\label{def:distribution-level-effect}
If \(\mathcal Z_{\mathrm{dist}}=\mathcal B\) is a Banach space, define
\(\Delta_{\mathrm{dist}} = T_{\mathrm{dist}}(P_Y^1)
- T_{\mathrm{dist}}(P_Y^0) \in\mathcal B\). If the target is only metric, define instead
\(\delta_{\mathrm{dist}} = d_{\mathrm{dist}}\!\left(
T_{\mathrm{dist}}(P_Y^1), T_{\mathrm{dist}}(P_Y^0)
\right)\).
\end{definition}

The vector \(\Delta_{\mathrm{dist}}\) retains the direction of the contrast, whereas \(\delta_{\mathrm{dist}}\) records only its magnitude. Because \(T_{\mathrm{dist}}\) may be non-injective, \(\delta_{\mathrm{dist}}=0\) means equality of the selected representations, not necessarily equality of the interventional laws.

\subsection{Identification by the causal
\texorpdfstring{\(g\)}{g}-formula}
\label{subsec:distribution-identification}

Let \(O=(X,A,Y)\). For \(a\in\{0,1\}\), write
\(P_{a,x} := \mathcal L(Y^a\mid X=x)\) for the causal conditional law and recall the observed conditional
outcome kernel \(Q_a(x,\cdot) := \mathcal L(Y\mid A=a,X=x)\). The first kernel is defined \(P_X\)-almost everywhere, whereas the second is initially defined \(P_{X\mid A=a}\)-almost everywhere. Under strict positivity, \(Q_a\) may also be regarded as defined \(P_X\)-almost everywhere.

\begin{theorem}[Identification of the interventional outcome law]
\label{thm:identification-interventional-law}
Assume consistency, conditional exchangeability, and strict positivity. Then
\[
P_Y^a
=
\int_{\mathcal X}Q_a(x,\cdot)\,dP_X(x),
\qquad a\in\{0,1\}.
\]
Equivalently, for every bounded measurable
\(h:\mathcal Y\to\mathbb R\),
\[
\int_{\mathcal Y}h(y)\,dP_Y^a(y)
=
\mathbb E\left[
\mathbb E\{h(Y)\mid A=a,X\}
\right].
\]
\end{theorem}

This is the standard law-valued causal \(g\)-formula
\citep{robins1986,hernanrobins2020}; no new causal identification principle is required.

\begin{corollary}[Identification of distribution-level topological objects]
\label{cor:identification-distribution-tcda}
If \(T_{\mathrm{dist}}\) is known and well defined at \(P_Y^0\) and \(P_Y^1\), then \(\Theta_{\mathrm{dist}}^a\), and hence \(\Delta_{\mathrm{dist}}\) or \(\delta_{\mathrm{dist}}\), is identified.
\end{corollary}

Thus topology is applied after the causal identification step.
Continuity of \(T_{\mathrm{dist}}\) is not required for identification, but it is essential for stable estimation.

\paragraph{A within-stratum causal estimand.}
Recall that \(P_{a,x}\) denotes the treatment-\(a\)
potential-outcome law within covariate stratum \(x\). Unlike the marginal distribution-level effect, the following estimand applies the topological representation before averaging over \(X\).

\begin{definition}[Covariate-standardized within-stratum effect]
\label{def:tau-dist}
Assume that \(P_{a,x}\in\mathcal P^\star(\mathcal Y)\) for
\(a=0,1\) and \(P_X\)-almost every \(x\), and that the following integrand is measurable and integrable. Define
\[
\tau_{\mathrm{dist}}
=
\mathbb E_X\!\left[
d_{\mathrm{dist}}\!\left(
T_{\mathrm{dist}}(P_{1,X}),
T_{\mathrm{dist}}(P_{0,X})
\right)
\right].
\]
When the target is a Banach space,
\(\tau_{\mathrm{dist}} = \int_{\mathcal X}
\left\| T_{\mathrm{dist}}(P_{1,x})
- T_{\mathrm{dist}}(P_{0,x}) \right\|_{\mathcal B}
\,dP_X(x)\).
\end{definition}

\begin{proposition}[Identification of \(\tau_{\mathrm{dist}}\)]
\label{prop:identification-tau-dist}
Assume consistency, conditional exchangeability, and strict positivity. Then
\[
P_{a,x}
=
\mathcal L(Y^a\mid X=x)
=
Q_a(x,\cdot)
\]
for \(a=0,1\) and \(P_X\)-almost every \(x\). Consequently,
\[
\tau_{\mathrm{dist}}
=
\mathbb E_X\!\left[
d_{\mathrm{dist}}\!\left(
T_{\mathrm{dist}}(Q_1(X,\cdot)),
T_{\mathrm{dist}}(Q_0(X,\cdot))
\right)
\right],
\]
and hence \(\tau_{\mathrm{dist}}\) is identified from the observational law.
\end{proposition}

This is the conditional version of the \(g\)-formula. Unlike the marginal interventional law, \(\tau_{\mathrm{dist}}\) depends on the covariates used to define the strata. Those covariates are therefore part of the scientific specification of the estimand.

\begin{remark}[Mix first or transform first]
\label{rem:mix-first}
Suppose that the target is a separable Banach space and that the map \(x \longmapsto T_{\mathrm{dist}}(P_{1,x})
- T_{\mathrm{dist}}(P_{0,x})\) is Bochner integrable. Define the signed within-stratum contrast
\[
\Delta_{\mathrm{strat}}
=
\int_{\mathcal X}
\left\{
T_{\mathrm{dist}}(P_{1,x})
-
T_{\mathrm{dist}}(P_{0,x})
\right\}
\,dP_X(x).
\]
In contrast, because \(P_Y^a = \int_{\mathcal X}P_{a,x}\,dP_X(x), \ a\in\{0,1\}\), the marginal distribution-level contrast is
\[
\Delta_{\mathrm{dist}}
=
T_{\mathrm{dist}}\!\left(
\int_{\mathcal X}P_{1,x}\,dP_X(x)
\right)
-
T_{\mathrm{dist}}\!\left(
\int_{\mathcal X}P_{0,x}\,dP_X(x)
\right).
\]
Under conditional exchangeability, the causal conditional laws
\(P_{a,x}\) in these formulas may be replaced by the observed kernels \(Q_a(x,\cdot)\).

The contrasts \(\Delta_{\mathrm{strat}}\) and
\(\Delta_{\mathrm{dist}}\) agree whenever \(T_{\mathrm{dist}}\) commutes with the two \(P_X\)-mixtures. This holds, for example, when \(T_{\mathrm{dist}}\) is barycentrically affine. Ordinary finite affinity also suffices under continuity and integrability conditions
that permit passage from finite mixtures to the corresponding
integrals. Moreover,
\(\|\Delta_{\mathrm{strat}}\|_{\mathcal B} \leq \tau_{\mathrm{dist}}\) by the Bochner--Jensen inequality. The inequality may be strict because signed within-stratum contrasts can cancel, whereas their norms cannot.
\end{remark}

\subsection{Inverse-probability representation and plug-in estimation}
\label{subsec:ipw-interventional-laws}

\begin{proposition}[Inverse-probability representation]
\label{prop:ipw-law}
For each \(a\in\{0,1\}\), under consistency, conditional
exchangeability, and strict positivity,
\[
\int_{\mathcal Y}h(y)\,dP_Y^a(y)
=
\mathbb E\left[
\frac{\mathbf 1\{A=a\}}{e_a(X)}h(Y)
\right]
\]
for every bounded measurable \(h:\mathcal Y\to\mathbb R\), where \(e_a(X)=\mathbb P(A=a\mid X)\).
\end{proposition}

This is the usual IPW identity applied to all bounded measurable test functions, and therefore characterizes the entire interventional law.

Given independent and identically distributed observations
\((X_i,A_i,Y_i)_{i=1}^n\), suppose that the fitted treatment probabilities satisfy \(0<\widehat e_a(X_i)<1\). The unnormalized Horvitz--Thompson measure
\[
\widehat\nu_{a,n}^{\mathrm{HT}}
=
\frac1n
\sum_{i=1}^n
\frac{\mathbf 1\{A_i=a\}}
     {\widehat e_a(X_i)}
\delta_{Y_i}
\]
has nonnegative weights but need not have total mass one. Define
\[
S_{a,n}
=
\sum_{i=1}^n
\frac{\mathbf 1\{A_i=a\}}
     {\widehat e_a(X_i)}.
\]
On the event \(\{S_{a,n}>0\}\), the normalized Hájek empirical law is
\[
\widetilde P_{Y,n}^{a,\mathrm{IPW}}
=
\frac{1}{S_{a,n}}
\sum_{i=1}^n
\frac{\mathbf 1\{A_i=a\}}
     {\widehat e_a(X_i)}
\delta_{Y_i}.
\]
Because strict positivity implies \(P(A=a)>0\), the probability of observing no units in treatment arm \(a\), and hence of \(S_{a,n}=0\), converges to zero.

If \(\mathcal Y\) is Polish and the true propensity score is used, the weighted strong law, applied to a countable convergence-determining class of bounded continuous functions, gives
\[
\widetilde P_{Y,n}^{a,\mathrm{IPW}}
\Longrightarrow
P_Y^a
\qquad\text{almost surely}.
\]
With estimated propensity scores, the analogous conclusion requires additional conditions ensuring convergence of the weighted empirical process, such as suitable consistency, boundedness, and sample-splitting or cross-fitting conditions. It does not follow from pointwise consistency of the propensity estimator alone.

More generally, let
\(\widehat P_{Y,n}^a\in\mathcal P^\star(\mathcal Y)\) be any probability-valued estimator of \(P_Y^a\). When the target is a Banach space, define
\(\widehat\Delta_{\mathrm{dist}} = T_{\mathrm{dist}}(\widehat P_{Y,n}^1) - T_{\mathrm{dist}}(\widehat P_{Y,n}^0)\). If, for \(a=0,1\), \(d_{\mathcal P}(\widehat P_{Y,n}^a,P_Y^a) \xrightarrow{p}0 \) and \(T_{\mathrm{dist}}\) is continuous from \((\mathcal P^\star(\mathcal Y),d_{\mathcal P})\) into
\((\mathcal B,\|\cdot\|_{\mathcal B})\), then the continuous mapping theorem gives \(\widehat\Delta_{\mathrm{dist}} \xrightarrow{p} \Delta_{\mathrm{dist}}\). For a metric-valued representation, the analogous plug-in estimator is
\[
\widehat\delta_{\mathrm{dist}}
=
d_{\mathrm{dist}}\!\left(
T_{\mathrm{dist}}(\widehat P_{Y,n}^1),
T_{\mathrm{dist}}(\widehat P_{Y,n}^0)
\right),
\]
and the same continuity argument gives
\(\widehat\delta_{\mathrm{dist}}\xrightarrow{p}\delta_{\mathrm{dist}}\).

Weak convergence alone may be insufficient. The DTM stability argument used in this paper requires \(W_2\)-consistency; support persistence requires support recovery; and density-level persistence requires convergence of the corresponding density estimator. Likewise, an augmented estimator of a distribution may be a signed measure and must be projected or otherwise constrained
before it can be inserted into \(T_{\mathrm{dist}}\).

\subsection{Non-reducibility and non-commutation}
\label{subsec:noncommutation}

Define the mean outcome-level functional \(\mathcal A(P) =\int_{\mathcal Y}T_{\mathrm{out}}(y)\,dP(y)\), on a convex class \(\mathcal Q\) of laws under which the integral exists. By linearity of integration, \(\mathcal A\) is affine:
\[
\mathcal A(\lambda P+(1-\lambda)Q)
=
\lambda\mathcal A(P)+(1-\lambda)\mathcal A(Q).
\]
Moreover, \(\mathbb E[T_{\mathrm{out}}(Y^a)] = \mathcal A(P_Y^a)\).

For a nonempty space \(K\) with finite zeroth Betti number, write \(\widetilde\beta_0(K)=\beta_0(K)-1\) for its reduced zeroth Betti number.

By contrast, a distribution-level map need not be affine. For example, if
\(T_{\mathrm{dist}}(P) = \widetilde\beta_0(\operatorname{supp}P)\), then \(T_{\mathrm{dist}}(\delta_{-1})=
T_{\mathrm{dist}}(\delta_1)=0\), whereas
\(T_{\mathrm{dist}}\left( \frac12\delta_{-1}+\frac12\delta_1
\right)=1\). Thus this support-based representation does not commute with mixing.

\begin{theorem}[Characterization of agreement between the two levels]
\label{thm:noncommutation-characterization}
Let
\(\mathcal A,T_{\mathrm{dist}}:\mathcal Q\to\mathcal B\), where \(\mathcal Q\) is nonempty. Then
\[
\mathcal A(P^1)-\mathcal A(P^0)
=
T_{\mathrm{dist}}(P^1)-T_{\mathrm{dist}}(P^0)
\]
for every \(P^0,P^1\in\mathcal Q\) if and only if
\(T_{\mathrm{dist}}-\mathcal A\) is constant on \(\mathcal Q\).

Consequently, if \(\mathcal Q\) is convex, \(\mathcal A\) is the mean functional above, and \(T_{\mathrm{dist}}\) is non-affine, then there exist \(P^0,P^1\in\mathcal Q\) for which the two contrasts differ.
\end{theorem}

\begin{proof}
Equality of the two contrasts is equivalent to
\[
T_{\mathrm{dist}}(P^1)-\mathcal A(P^1)
=
T_{\mathrm{dist}}(P^0)-\mathcal A(P^0).
\]
This holds for every pair \(P^0,P^1\) precisely when
\(T_{\mathrm{dist}}-\mathcal A\) is constant.

For the final statement, suppose that the contrasts agreed for every pair. Then
\[
T_{\mathrm{dist}}(P)=\mathcal A(P)+c
\]
for some fixed \(c\in\mathcal B\) and every \(P\in\mathcal Q\). Because \(\mathcal A\) is affine, so is \(\mathcal A+c\), contradicting the assumed non-affineness of \(T_{\mathrm{dist}}\).
\end{proof}

The theorem separates two issues. On a convex class of probability laws, failure to commute with mixtures is equivalent to non-affineness of \(T_{\mathrm{dist}}\). Failure of the outcome- and distribution-level contrasts to agree is characterized
more generally by nonconstancy of \(T_{\mathrm{dist}}-\mathcal A\). Two distinct affine maps can therefore also give different contrasts.

\begin{example}[Non-degenerate divergence of the two levels]
\label{ex:nondegenerate-noncommutation}
Let
\[
K=\{(1,0),(0,1),(-1,0),(0,-1)\},
\qquad
r_0=1,\quad r_1=2,
\]
and choose \(R>2\). For \(\theta\in[0,2\pi)\), let \(\rho_\theta\) denote counterclockwise rotation of \(\mathbb R^2\) through angle \(\theta\). Let \(C\) be uniform on the circle of radius \(R\), let \(\Theta\) be independently uniform on \([0,2\pi)\), and define \(Y^a=C+r_a\rho_\Theta K\). Thus each potential outcome is a randomly translated and rotated square.

Let \(T_{\mathrm{out}}(S)\) be the total degree-one persistence of the Vietoris--Rips filtration of \(S\). Under the diameter-threshold convention for the Vietoris--Rips filtration, the square \(K\) has one degree-one class born at \(\sqrt2\) and dying at \(2\), so \(T_{\mathrm{out}}(Y^a) = r_a(2-\sqrt2)\). Therefore
\[
\operatorname{TATE}_{\mathrm{out}}
=
(r_1-r_0)(2-\sqrt2)
=
2-\sqrt2.
\]

For a four-point set \(S\), let \(\kappa_S=\frac14\sum_{z\in S}\delta_z\), and for a law \(P\) on such sets define
\[
\mu_P=\int\kappa_S\,dP(S),
\qquad
T_{\mathrm{dist}}(P)
=
\beta_1(\operatorname{supp}\mu_P).
\]
For \(P=P_Y^a\), the support of \(\mu_P\) is the annulus
\(\{x\in\mathbb R^2:R-r_a\leq\|x\|\leq R+r_a\}\). It therefore has first Betti number one for both treatment levels. Hence \(T_{\mathrm{dist}}(P_Y^0) = T_{\mathrm{dist}}(P_Y^1) = 1,\ \Delta_{\mathrm{dist}}=0\). Both representations are nontrivial and take values in \(\mathbb R\), yet their causal contrasts differ.
\end{example}

\subsection{A zero classical ATE with nonzero distributional topology}
\label{subsec:zero-ate-nonzero-topology}

\begin{example}[Zero mean effect but nonzero density topology]
\label{ex:zero-ate-density-topology}
Fix \(\rho>0\) and let
\[
Y^0
\sim
\frac12N((-1,0),\sigma^2I_2)
+
\frac12N((1,0),\sigma^2I_2),
\qquad
Y^1\sim N((0,0),\rho^2I_2).
\]
Both laws have mean zero, so \(\mathbb E[Y^1]-\mathbb E[Y^0]=0\).

Let \(f_0,f_1\) be their densities, and write
\[
s_\sigma=f_0(0,0),
\qquad
m_\sigma=f_0(1,0),
\qquad
M_\rho=f_1(0,0).
\]
As \(\sigma\downarrow0\), \(s_\sigma\longrightarrow0, \ m_\sigma\longrightarrow\infty\). Hence, for sufficiently small \(\sigma\), one may choose \(s_\sigma<\lambda<\min\{m_\sigma,M_\rho\}\).

The treatment superlevel set
\[
L_\lambda(P_Y^1)
=
\{y:f_1(y)\geq\lambda\}
\]
is a nonempty closed disk and is therefore connected. For the control law,
\[
f_0(0,x_2)
=
s_\sigma
\exp\left(-\frac{x_2^2}{2\sigma^2}\right)
<
\lambda
\]
for every \(x_2\), whereas \(f_0(-1,0)=f_0(1,0)=m_\sigma>\lambda\). Thus the control superlevel set contains points on both sides of the
line \(x_1=0\) but does not intersect that line, and is disconnected.

Consequently,
\(T_{\mathrm{dist},\lambda}(P) = \widetilde\beta_0(L_\lambda(P))\)
satisfies
\(T_{\mathrm{dist},\lambda}(P_Y^1)=0,
\  T_{\mathrm{dist},\lambda}(P_Y^0)\geq1\). The classical ATE is therefore zero while the distribution-level topological contrast is nonzero.
\end{example}

Both Gaussian laws have support \(\mathbb R^2\), so this distinction cannot be detected by support topology. It arises from density superlevel sets and therefore illustrates why the filtration must be chosen according to the effective geometry of scientific interest.

\subsection{Distance-to-measure distribution-level TCDA}
\label{subsec:dtm-distribution-tcda}

A robust alternative to support topology is provided by the
distance-to-measure construction \citep{chazalcohensteinermrigot2011,buchetetal2016}. Let
\((\mathcal Y,d_{\mathcal Y})\) be a Polish metric space and \(P\) a Borel probability measure. For \(\ell\in(0,1)\), define
\[
\delta_{P,\ell}(y)
=
\inf\left\{
r>0:P(\overline B(y,r))>\ell
\right\}.
\]
For \(m\in(0,1)\), the distance-to-measure function is
\(d_{P,m}(y) = \left(
\frac1m \int_0^m \delta_{P,\ell}(y)^2\,d\ell
\right)^{1/2}\). It is \(1\)-Lipschitz in \(y\) and incorporates local probability mass rather than only the full support.

The DTM sublevel filtration is
\(\mathcal F_{\mathrm{DTM}}^m(P) = \bigl\{
\{y:d_{P,m}(y)\leq t\} \bigr\}_{t\geq0}\). Whenever its degree-\(k\) persistence module is \(q\)-tame, write
\(D_{k,m}^{\mathrm{DTM}}(P) = \operatorname{Dgm}_k(
\mathcal F_{\mathrm{DTM}}^m(P))\). For a vectorization \(\Phi\), define \(T_{\mathrm{DTM}}^{m,k,\Phi}(P) =
\Phi(D_{k,m}^{\mathrm{DTM}}(P))\), and the corresponding causal effect by
\[
\Delta_{\mathrm{DTM}}^{m,k,\Phi}
=
T_{\mathrm{DTM}}^{m,k,\Phi}(P_Y^1)
-
T_{\mathrm{DTM}}^{m,k,\Phi}(P_Y^0).
\]

For \(P,Q\in\mathcal P_2(\mathcal Y)\), the standard DTM stability inequality is
\[
\|d_{P,m}-d_{Q,m}\|_\infty
\leq
m^{-1/2}W_2(P,Q).
\]
Under the standing persistence regularity conditions,
\[
d_B\left(
D_{k,m}^{\mathrm{DTM}}(P),
D_{k,m}^{\mathrm{DTM}}(Q)
\right)
\leq
m^{-1/2}W_2(P,Q).
\]
If \(\Phi\) is \(L_\Phi\)-Lipschitz in bottleneck distance, then
\[
\left\|
T_{\mathrm{DTM}}^{m,k,\Phi}(P)
-
T_{\mathrm{DTM}}^{m,k,\Phi}(Q)
\right\|_{\mathcal B}
\leq
\frac{L_\Phi}{\sqrt m}W_2(P,Q).
\]
The general transfer of this stability bound to causal contrasts is developed in Section~\ref{sec:stability}.

The mass parameter \(m\) is part of the estimand. Smaller values retain more local structure but increase the stability constant \(m^{-1/2}\); larger values provide greater smoothing. The outcome metric \(d_{\mathcal Y}\), homological degree \(k\), and vectorization \(\Phi\) are likewise substantive modelling choices.

Empirically, \(P_Y^a\) is replaced by a probability-valued estimator
\(\widehat P_{Y,n}^a\in\mathcal P_2(\mathcal Y)\), such as a normalized Hájek law or a probability-valued \(g\)-formula estimator. An unnormalized or signed measure cannot be used directly. Moreover, weak convergence alone does not guarantee DTM consistency: the stability bound requires
\(W_2(\widehat P_{Y,n}^a,P_Y^a)\xrightarrow{p}0\). Under this condition, the DTM functions, persistence diagrams, and stable vectorizations converge in their respective metrics.

Outcome-level TCDA therefore averages topological representations of individual potential outcomes, whereas distribution-level TCDA transforms the interventional laws themselves. Neither construction generally determines the other, and the order of covariate mixing and topological transformation is part of the causal question.

%%%%%%%%%%%%%%%%%%%%%%%%%%%%%%%%%%%%%%%%%%%%%%%%%%%%%%%%
%%%%%%%%%%%%%%%%%%%%%%%%%%%%%%%%%%%%%%%%%%%%%%%%%%%%%%%%

\section{Stability of topological causal estimands}
\label{sec:stability}

Identification and stability answer different questions. Identification asks whether a causal estimand is determined by the observational law under specified causal assumptions. Stability asks how much that estimand changes when its potential outcomes, interventional laws, or topological representations are perturbed.

The principle of this section is simple: if the topological
representation is Lipschitz, then the causal contrast constructed from it inherits the same stability, up to contributions from the two treatment arms. These results concern perturbations of the causal objects themselves; they do not provide robustness to violations of exchangeability, positivity, or consistency.

\subsection{Abstract stability transfer}
\label{subsec:abstract-stability-transfer}

Let \((\mathcal Y,d_{\mathcal Y})\) be a metric outcome space and \(T_{\mathrm{out}}:\mathcal Y\to\mathcal B\) an outcome-level representation into a separable Banach space.

\begin{assumption}[Lipschitz outcome-level representation]
\label{ass:lipschitz-outcome-T}
There exists \(L_{\mathrm{out}}\geq0\) such that
\[
\|T_{\mathrm{out}}(y)-T_{\mathrm{out}}(y')\|_{\mathcal B}
\leq
L_{\mathrm{out}}d_{\mathcal Y}(y,y')
\]
for all \(y,y'\in\mathcal Y\).
\end{assumption}

\begin{theorem}[Coupling stability of the outcome-level effect]
\label{thm:stability-outcome-tate}
Let \(M\) and \(M'\) be two causal models, and for each
\(a\in\{0,1\}\) let \((Y^a,\widetilde Y^a)\) be a coupling of their treatment-\(a\) potential-outcome laws. Under
Assumption~\ref{ass:lipschitz-outcome-T}, if the represented outcomes are Bochner integrable and
\(\mathbb E[d_{\mathcal Y}(Y^a,\widetilde Y^a)]<\infty\), then
\[
\left\|
\operatorname{TATE}_{\mathrm{out}}(M)
-
\operatorname{TATE}_{\mathrm{out}}(M')
\right\|_{\mathcal B}
\leq
L_{\mathrm{out}}
\sum_{a\in\{0,1\}}
\mathbb E[
d_{\mathcal Y}(Y^a,\widetilde Y^a)
].
\]
\end{theorem}

\begin{proof}
Linearity of the Bochner integral, the triangle inequality, and Bochner--Jensen give
\[
\left\|
\operatorname{TATE}_{\mathrm{out}}(M)
-
\operatorname{TATE}_{\mathrm{out}}(M')
\right\|_{\mathcal B}
\leq
\sum_{a\in\{0,1\}}
\mathbb E\left[
\|T_{\mathrm{out}}(Y^a)
-
T_{\mathrm{out}}(\widetilde Y^a)\|_{\mathcal B}
\right].
\]
Applying Assumption~\ref{ass:lipschitz-outcome-T} to each term proves the result.
\end{proof}

\begin{corollary}[Wasserstein stability]
\label{cor:wasserstein-outcome-stability}
Let \(P_{Y,M}^a\) and \(P_{Y,M'}^a\) be the treatment-\(a\) interventional outcome laws, and assume that they belong to
\(\mathcal P_1(\mathcal Y)\). Then
\[
\left\|
\operatorname{TATE}_{\mathrm{out}}(M)
-
\operatorname{TATE}_{\mathrm{out}}(M')
\right\|_{\mathcal B}
\leq
L_{\mathrm{out}}
\sum_{a\in\{0,1\}}
W_1(P_{Y,M}^a, P_{Y,M'}^a).
\]
\end{corollary}

This follows by taking the infimum in
Theorem~\ref{thm:stability-outcome-tate} over all treatment-specific couplings.

\smallbreak 

For distribution-level effects, let
\(T_{\mathrm{dist}}:
(\mathcal P_\star(\mathcal Y),d_{\mathcal P})
\longrightarrow
(\mathcal Z_{\mathrm{dist}},d_{\mathrm{dist}})
\) be a representation of probability laws.

\begin{assumption}[Lipschitz distribution-level representation]
\label{ass:lipschitz-distribution-T}
There exists \(L_{\mathrm{dist}}\geq0\) such that
\[
d_{\mathrm{dist}}\left(
T_{\mathrm{dist}}(P),
T_{\mathrm{dist}}(Q)
\right)
\leq
L_{\mathrm{dist}}d_{\mathcal P}(P,Q)
\]
for all \(P,Q\in\mathcal P_\star(\mathcal Y)\).
\end{assumption}

\begin{theorem}[Distribution-level stability]
\label{thm:distribution-stability}
Let \(M\) and \(M'\) have interventional outcome laws in
\(\mathcal P_\star(\mathcal Y)\). Define
\(\delta_{\mathrm{dist}}(M) = d_{\mathrm{dist}}\left(
T_{\mathrm{dist}}(P_{Y,M}^1), T_{\mathrm{dist}}(P_{Y,M}^0)
\right)\). Then
\[
\left|
\delta_{\mathrm{dist}}(M)
-
\delta_{\mathrm{dist}}(M')
\right|
\leq
L_{\mathrm{dist}}
\sum_{a\in\{0,1\}}
d_{\mathcal P}(P_{Y,M}^a,P_{Y,M'}^a).
\]

If \(\mathcal Z_{\mathrm{dist}}=\mathcal B\) is a Banach space
equipped with its norm-induced metric, \(d_{\mathrm{dist}}(u,v)=\|u-v\|_{\mathcal B}\), define
\(\Delta_{\mathrm{dist}}(M) = T_{\mathrm{dist}}(P_{Y,M}^1)
- T_{\mathrm{dist}}(P_{Y,M}^0)\). Then
\[
\left\|
\Delta_{\mathrm{dist}}(M)
-
\Delta_{\mathrm{dist}}(M')
\right\|_{\mathcal B}
\leq
L_{\mathrm{dist}}
\sum_{a\in\{0,1\}}
d_{\mathcal P}(P_{Y,M}^a,P_{Y,M'}^a).
\]
\end{theorem}

\begin{proof}
For the metric-valued effect, the four-point inequality gives
\[
\left|
d_{\mathrm{dist}}(\Theta_{1,M},\Theta_{0,M})
-
d_{\mathrm{dist}}(\Theta_{1,M'},\Theta_{0,M'})
\right|
\leq
d_{\mathrm{dist}}(\Theta_{1,M},\Theta_{1,M'})
+
d_{\mathrm{dist}}(\Theta_{0,M},\Theta_{0,M'}),
\]
where
\(\Theta_{a,M}=T_{\mathrm{dist}}(P_{Y,M}^a)\).
Assumption~\ref{ass:lipschitz-distribution-T} gives the first claim. The Banach-space statement follows from the ordinary triangle inequality applied to the two treatment-specific differences.
\end{proof}

\subsection{Matching the diagram metric to the vectorization}
\label{subsec:diagram-metric}

A persistence pipeline has two stability steps:
\(u \longmapsto D(u) \longmapsto \Phi(D(u))\). The metric controlling the diagram map must match the metric in which \(\Phi\) is stable.

\begin{definition}[Diagram-stable vectorization]
\label{def:diagram-stable-summary}
Let \(d_{\mathsf D}\) be a metric on a diagram class \(\mathcal D\). A map \(\Phi:\mathcal D\to\mathcal B\) is
\(d_{\mathsf D}\)-stable with constant \(L_\Phi\) if
\[
\|\Phi(D)-\Phi(D')\|_{\mathcal B}
\leq
L_\Phi d_{\mathsf D}(D,D')
\]
for all \(D,D'\in\mathcal D\).
\end{definition}

If \(d_{\mathsf D}(D(u),D(u')) \leq C_Dd_{\mathcal U}(u,u')\), then \(T=\Phi\circ D\) is \(L_\Phi C_D\)-Lipschitz.

When diagram stability is known in bottleneck distance but
\(\Phi\) is stable in \(W_p\), the following finite-cardinality comparison is useful. Throughout this comparison, \(d_B\) and \(W_p\) are defined using the same ground metric on the birth--death plane, taken here to be the \(\ell^\infty\) metric.

\begin{lemma}[Bottleneck--Wasserstein comparison]
\label{lem:bottleneck-wasserstein}
If \(D\) and \(D'\) each have at most \(N\) off-diagonal points, then
\[
W_p(D,D')
\leq
(2N)^{1/p}d_B(D,D'),
\qquad 1\leq p<\infty.
\]
\end{lemma}

\begin{proof}
For every \(\varepsilon>d_B(D,D')\), choose a bottleneck matching of cost at most \(\varepsilon\). It has at most
\(|D|+|D'|\leq2N\) nonzero contributions, each at most
\(\varepsilon^p\). Hence \(W_p(D,D')^p\leq2N\varepsilon^p\). Letting \(\varepsilon\downarrow d_B(D,D')\) proves the claim.
\end{proof}

Therefore, if \(\Phi\) is \(L_{\Phi,p}\)-Lipschitz in \(W_p\), then on this diagram class it is bottleneck-Lipschitz with constant \(\kappa_\Phi = L_{\Phi,p}(2N)^{1/p}\). A bounded-cardinality assumption should not be replaced merely by bounded total persistence. Direct Wasserstein stability theorems may be used instead, but their geometric and total-persistence assumptions must be checked explicitly
\citep{cohensteineredelsbrunnerharermileyko2010,skrabaturner2020}.

Persistence landscapes are directly stable from bottleneck distance to the appropriate supremum norm. Betti curves, persistence images, and other summaries are commonly stable in Wasserstein-type metrics under summary-specific target norms and regularity conditions. Thus a statement that a vectorization is ``stable'' is incomplete unless both
the diagram metric and the target norm are specified.

\paragraph{Power-weighted silhouettes.}
Kim and Lee~\citep[Theorem~5.3]{kimlee2026} establish a bound
\[
\|\phi_D-\phi_{D'}\|_\infty
\leq
C_{\mathrm{KL}}W_1(D,D')
\]
for power-weighted silhouettes under their boundedness and
normalization conditions. We retain their constant in the source form \(C_{\mathrm{KL}}\), since its precise expression depends on the conditions and auxiliary constant used in their theorem. For \(r=1\), their bound gives \(C_{\mathrm{KL}}=3\); for \(r>1\), additional
control preventing arbitrarily small positive persistences is required.

\begin{proposition}[Silhouette metric conversion]
\label{prop:silhouette-constant}
If, in addition, every diagram has at most \(N\)
off-diagonal points, then \(\|\phi_D-\phi_{D'}\|_\infty \leq 2NC_{\mathrm{KL}}\,d_B(D,D')\).
\end{proposition}

\begin{proof}
Lemma~\ref{lem:bottleneck-wasserstein}, with \(p=1\), gives
\(W_1(D,D')\leq2N\,d_B(D,D')\). Combining this with the
Kim--Lee bound
\(\|\phi_D-\phi_{D'}\|_\infty
\leq C_{\mathrm{KL}}W_1(D,D')\)
proves the result.
\end{proof}

This is a metric conversion of the Kim--Lee result, not a new
silhouette-stability theorem.

For reference, Table~\ref{tab:stability-transfer} summarizes the classical diagram-level bounds and the resulting TCDA Lipschitz constants.

\begin{table}[ht]
\centering
\begin{tabular}{@{}llll@{}}
\hline
Filtration
& Input distance
& \(C_D\)
& \(L_T=\kappa_\Phi C_D\) \\
\hline
Sublevel sets
& \(\|f-g\|_\infty\)
& \(1\)
& \(\kappa_\Phi\) \\

Vietoris--Rips
& \(d_{GH}(K,L)\)
& \(2\)
& \(2\kappa_\Phi\) \\

Distance-to-measure
& \(W_2(P,Q)\)
& \(m^{-1/2}\)
& \(\kappa_\Phi m^{-1/2}\) \\
\hline
\end{tabular}
\caption{Stability constants, where
\(d_B(D(u),D(u'))\leq C_Dd_{\mathcal U}(u,u')\) and
\(T=\Phi\circ D\) is \(L_T\)-Lipschitz.}
\label{tab:stability-transfer}
\end{table}

\subsection{Sublevel-set persistence}
\label{subsec:sublevel-stability}

Suppose that each outcome \(y\) determines a tame function
\(f_y:\Omega\to\mathbb R\) on a fixed triangulable space, and let \(D_k(y)\) be the persistence diagram of its sublevel-set filtration.

\begin{lemma}[Classical sublevel-set stability]
\label{lem:sublevel-diagram-stability}
Under the usual tameness assumptions,
\[
d_B(D_k(y),D_k(y'))
\leq
\|f_y-f_{y'}\|_\infty.
\]
\end{lemma}

This is the classical bottleneck stability theorem
\citep{cohensteineredelsbrunnerharer2007} and no proof is repeated here.

\begin{proposition}[Sublevel-set TCDA stability]
\label{prop:sublevel-causal-stability}
Let \(\Phi\) be bottleneck-Lipschitz with constant \(\kappa_\Phi\), and define \(T_{\mathrm{out}}(y)=\Phi(D_k(y))\). Then
\[
\|T_{\mathrm{out}}(y)-T_{\mathrm{out}}(y')\|_{\mathcal B}
\leq
\kappa_\Phi\|f_y-f_{y'}\|_\infty.
\]
Consequently, for causal models \(M,M'\) and
treatment-specific couplings
\((Y^a,\widetilde Y^a)\), \(a\in\{0,1\}\), satisfying the
integrability conditions of
Theorem~\ref{thm:stability-outcome-tate},
\[
\left\|
\operatorname{TATE}_{\mathrm{out}}(M)
-
\operatorname{TATE}_{\mathrm{out}}(M')
\right\|_{\mathcal B}
\leq
\kappa_\Phi
\sum_{a\in\{0,1\}}
\mathbb E\|f_{Y^a}-f_{\widetilde Y^a}\|_\infty.
\]
\end{proposition}

\begin{proof}
The first bound follows by composing
Lemma~\ref{lem:sublevel-diagram-stability} with the
\(\kappa_\Phi\)-Lipschitz map \(\Phi\). Applying the triangle
inequality and Bochner--Jensen to the two treatment arms, as in Theorem~\ref{thm:stability-outcome-tate}, and then using the first bound gives the second inequality.
\end{proof}

\subsection{Vietoris--Rips persistence}
\label{subsec:vr-stability}

Suppose that each outcome \(y\) determines a nonempty compact
metric space \(K_y\), and let \(D_k^{\mathrm{VR}}(K_y)\) denote its Vietoris--Rips persistence diagram. Under the diameter-threshold convention, totally bounded metric spaces have \(q\)-tame Vietoris--Rips persistence, and
\(d_B\left( D_k^{\mathrm{VR}}(K), D_k^{\mathrm{VR}}(L) \right) \leq 2d_{GH}(K,L)\) \citep{chazaldesilvaoudot2014}. For the Hausdorff formulation
below, assume additionally that all \(K_y\) are compact subsets of a common metric space \((E,d_E)\); then
\(d_{GH}(K,L)\leq d_H^E(K,L)\).

\begin{lemma}[Vietoris--Rips stability]
\label{lem:vr-stability}
For compact subsets \(K,L\subseteq E\),
\[
d_B\left(
D_k^{\mathrm{VR}}(K),
D_k^{\mathrm{VR}}(L)
\right)
\leq
2d_H^E(K,L).
\]
\end{lemma}

The constant \(2\) depends on the filtration convention.

\begin{proposition}[Vietoris--Rips TCDA stability]
\label{prop:vr-causal-stability}
Let \(\Phi\) be bottleneck-Lipschitz with constant \(\kappa_\Phi\), and define
\(T_{\mathrm{out}}(y) = \Phi(D_k^{\mathrm{VR}}(K_y))\). Then
\[
\|T_{\mathrm{out}}(y)-T_{\mathrm{out}}(y')\|_{\mathcal B}
\leq
2\kappa_\Phi d_H^E(K_y,K_{y'}).
\]
Consequently, for causal models \(M,M'\) and
treatment-specific couplings
\((Y^a,\widetilde Y^a)\), \(a\in\{0,1\}\), satisfying the
integrability conditions of
Theorem~\ref{thm:stability-outcome-tate},
\[
\left\|
\operatorname{TATE}_{\mathrm{out}}(M)
-
\operatorname{TATE}_{\mathrm{out}}(M')
\right\|_{\mathcal B}
\leq
2\kappa_\Phi
\sum_{a\in\{0,1\}}
\mathbb E[
d_H^E(K_{Y^a},K_{\widetilde Y^a})
].
\]
\end{proposition}

\begin{proof}
The first bound follows by composing
Lemma~\ref{lem:vr-stability} with the
\(\kappa_\Phi\)-Lipschitz map \(\Phi\). Applying this bound to
the coupled potential outcomes in each treatment arm and using
the triangle inequality and Bochner--Jensen gives the second
inequality.
\end{proof}

For arbitrary compact metric spaces, \(d_H^E\) may be replaced by \(d_{GH}\). If \(K_y\) is itself estimated, its Hausdorff or Gromov--Hausdorff convergence requires a separate sampling argument. Literal support topology may also be uninformative under unbounded noise, motivating the DTM construction below.

\subsection{Distance-to-measure persistence}
\label{subsec:dtm-stability}

Recall the DTM construction of
Section~\ref{subsec:dtm-distribution-tcda}. Let
\(P,Q\in\mathcal P_2(E)\), fix \(m\in(0,1)\), and assume that the corresponding degree-\(k\) DTM persistence modules are \(q\)-tame.

\begin{lemma}[DTM stability]
\label{lem:dtm-stability}
\[
d_B\left(
D_{k,m}^{\mathrm{DTM}}(P),
D_{k,m}^{\mathrm{DTM}}(Q)
\right)
\leq
\|d_{P,m}-d_{Q,m}\|_\infty
\leq
m^{-1/2}W_2(P,Q).
\]
\end{lemma}

The first inequality is sublevel-set stability, and the second is the standard DTM stability bound
\citep{chazalcohensteinermrigot2011,buchetetal2016}.

\begin{proposition}[DTM distribution-level stability]
\label{prop:dtm-causal-stability}
Let \(\Phi\) be bottleneck-Lipschitz with constant \(\kappa_\Phi\), and define
\(T_{\mathrm{DTM}}^{m,k,\Phi}(P) = \Phi(D_{k,m}^{\mathrm{DTM}}(P))\). Then
\[
\left\|
T_{\mathrm{DTM}}^{m,k,\Phi}(P)
-
T_{\mathrm{DTM}}^{m,k,\Phi}(Q)
\right\|_{\mathcal B}
\leq
\kappa_\Phi m^{-1/2}W_2(P,Q).
\]
For a causal model \(M\), define
\(\Delta_{\mathrm{DTM}}^{m,k,\Phi}(M) = T_{\mathrm{DTM}}^{m,k,\Phi}(P_{Y,M}^1) - T_{\mathrm{DTM}}^{m,k,\Phi}(P_{Y,M}^0)\). Then, for any two causal models \(M,M'\) whose interventional laws belong to \(\mathcal P_2(E)\),
\[
\left\|
\Delta_{\mathrm{DTM}}^{m,k,\Phi}(M)
-
\Delta_{\mathrm{DTM}}^{m,k,\Phi}(M')
\right\|_{\mathcal B}
\leq
\kappa_\Phi m^{-1/2}
\sum_{a\in\{0,1\}}
W_2(P_{Y,M}^a,P_{Y,M'}^a).
\]
\end{proposition}

\begin{proof}
The first bound follows by composing
Lemma~\ref{lem:dtm-stability} with the
\(\kappa_\Phi\)-Lipschitz map \(\Phi\). Applying that bound to
the treatment-specific pairs
\((P_{Y,M}^a,P_{Y,M'}^a)\), \(a=0,1\), and using the triangle
inequality proves the causal-effect bound.
\end{proof}

The parameter \(m\) is part of the estimand. Small values retain more local information but worsen the worst-case constant \(m^{-1/2}\). Weak convergence of laws is not sufficient for this bound: \(W_2\)-convergence also requires appropriate second-moment control.

\subsection{End-to-end plug-in bounds}
\label{subsec:end-to-end-distribution}

The abstract distribution-level bound immediately gives an
error-propagation result for any probability-valued estimator of the interventional laws.

\begin{proposition}[End-to-end plug-in bound]
\label{prop:end-to-end-distribution}
Suppose that
\(\mathcal Z_{\mathrm{dist}}=\mathcal B\) is a Banach space
equipped with its norm-induced metric and that
\(T_{\mathrm{dist}}\) satisfies
Assumption~\ref{ass:lipschitz-distribution-T}. Let
\(\widehat P^0,\widehat P^1\) be probability-valued estimators
taking values almost surely in
\(\mathcal P_\star(\mathcal Y)\), and define
\(\widehat\Delta_{\mathrm{dist}} = T_{\mathrm{dist}}(\widehat P^1) - T_{\mathrm{dist}}(\widehat P^0)\). Then
\[
\left\|
\widehat\Delta_{\mathrm{dist}}
-
\Delta_{\mathrm{dist}}
\right\|_{\mathcal B}
\leq
L_{\mathrm{dist}}
\sum_{a\in\{0,1\}}
d_{\mathcal P}(\widehat P^a,P_Y^a).
\]

For the DTM representation, assume additionally that
\(\widehat P^a\in\mathcal P_2(E)\) almost surely and that the
DTM diagrams of \(\widehat P^a\) satisfy the standing
\(q\)-tameness and diagram-class conditions. Define
\(\widehat\Delta_{\mathrm{DTM}}^{m,k,\Phi} = T_{\mathrm{DTM}}^{m,k,\Phi}(\widehat P^1) - T_{\mathrm{DTM}}^{m,k,\Phi}(\widehat P^0)\). Then
\[
\left\|
\widehat\Delta_{\mathrm{DTM}}^{m,k,\Phi}
-
\Delta_{\mathrm{DTM}}^{m,k,\Phi}
\right\|_{\mathcal B}
\leq
\kappa_\Phi m^{-1/2}
\sum_{a\in\{0,1\}}
W_2(\widehat P^a,P_Y^a).
\]
\end{proposition}

\begin{proof}
By the triangle inequality and Lipschitz continuity,
\[
\|\widehat\Delta_{\mathrm{dist}}
-\Delta_{\mathrm{dist}}\|_{\mathcal B}
\leq
\sum_{a\in\{0,1\}}
\|T_{\mathrm{dist}}(\widehat P^a)
-
T_{\mathrm{dist}}(P_Y^a)\|_{\mathcal B}
\leq
L_{\mathrm{dist}}
\sum_{a\in\{0,1\}}
d_{\mathcal P}(\widehat P^a,P_Y^a).
\]
The DTM statement follows by substituting \(L_{\mathrm{dist}}=\kappa_\Phi m^{-1/2}\).
\end{proof}

\begin{corollary}[Consistency and rate transfer]
\label{cor:consistency-distribution}
If \(d_{\mathcal P}(\widehat P^a,P_Y^a)
\xrightarrow{p}0,\ a=0,1\), then
\(\widehat\Delta_{\mathrm{dist}} \xrightarrow{p}
\Delta_{\mathrm{dist}}\). More generally, if
\(d_{\mathcal P}(\widehat P^a,P_Y^a)=O_p(r_{a,n})\), then
\[
\|\widehat\Delta_{\mathrm{dist}}
-
\Delta_{\mathrm{dist}}\|_{\mathcal B}
=
O_p\!\left(
L_{\mathrm{dist}}(r_{0,n}+r_{1,n})
\right).
\]
\end{corollary}

The corollary above assumes a fixed representation. If the
representation varies with \(n\), write its Lipschitz constant
as \(L_{\mathrm{dist},n}\) and its corresponding target as
\(\Delta_{\mathrm{dist},n}\). The same argument yields
plug-in consistency for this moving target provided
\(L_{\mathrm{dist},n} \sum_{a\in\{0,1\}}
d_{\mathcal P}(\widehat P^a,P_Y^a) \xrightarrow{p}0\). Convergence to a fixed limiting estimand additionally requires
\(\Delta_{\mathrm{dist},n}\to\Delta_{\mathrm{dist}}\). This
distinction matters when a cardinality bound \(N=N_n\) grows
or when the DTM parameter \(m=m_n\) tends to zero.

Finally, this proposition transfers an already established
\(d_{\mathcal P}\)-error bound. It does not prove that a particular \(g\)-formula, Hájek-weighted, smoothed, or projected estimator converges in that metric.

%%%%%%%%%%%%%%%%%%%%%%%%%%%%%%%%%%%%%%%%%%%%%%%%%%%%%%%%
%%%%%%%%%%%%%%%%%%%%%%%%%%%%%%%%%%%%%%%%%%%%%%%%%%%%%%%%

\section{Topology-assisted causal discovery: scope and limits}
\label{sec:topological-identifiability}

The preceding sections use topology to define causal effects after the causal target has been specified. A different question is whether topological information can help recover causal structure itself. This requires substantially stronger assumptions: geometric or topological structure in an observational distribution is not, by itself, causal
structure~\citep{spirtes2000,pearl2009}.

Topology-assisted discovery is nevertheless possible on restricted model classes. For example, a recent preprint uses a degree-zero persistent-homology functional of cross-fitted
regressor--residual clouds to infer direction in bivariate
additive-noise models~\citep{elbouchattaoui2026}. Our aim here is not to propose another general discovery algorithm, but to formalize the separation property that any such method requires. This use of topological summaries of data is distinct from the learning-theoretic use of topologies on spaces of structural causal models studied by
\citet{ibelingicard2021}.

Let \((\mathcal B,d_{\mathcal B})\) be a metric space of
topological summaries, let
\(\mathcal C\subseteq\bigcup_G\mathfrak M(G)\) be a class of
causal models, and let \(\sim\) denote the causal equivalence
relation of interest. Let \(T_{\mathrm{obs}}\) be a map from
the relevant observational laws into \(\mathcal B\), and write
\(\theta_T(M) = T_{\mathrm{obs}}(P_M^{\mathrm{obs}})
\in\mathcal B, \ M\in\mathcal C\). For an equivalence class \(A\in\mathcal C/{\sim}\), define its attainable summary set by
\(\Theta_T(A) =\{\theta_T(M):M\in A\}\).

\begin{definition}[\(T_{\mathrm{obs}}\)-identifiability]
\label{def:T-identifiability}
The equivalence classes of \(\mathcal C/{\sim}\) are
\emph{\(T_{\mathrm{obs}}\)-identifiable} if
\[
\Theta_T(A)\cap\Theta_T(B)=\varnothing
\qquad
\text{for every }A\neq B.
\]
\end{definition}

Thus \(T_{\mathrm{obs}}\)-identifiability means that the topological summary determines the relevant equivalence class, although it need not be constant within a class. If
\(\Theta_T(A)\cap\Theta_T(B)\neq\varnothing\) for two distinct classes, no procedure using only \(T_{\mathrm{obs}}\) can distinguish them uniformly. For statistical recovery, disjointness alone is insufficient.

\begin{definition}[Topological separation]
\label{def:topological-separation}
Two equivalence classes
\(A,B\in\mathcal C/{\sim}\) are \emph{topologically separated with margin \(\eta>0\)} if
\[
d_{\mathcal B}\bigl(\Theta_T(A),\Theta_T(B)\bigr)
:=
\inf_{\substack{M\in A\\M'\in B}}
d_{\mathcal B}\bigl(\theta_T(M),\theta_T(M')\bigr)
\geq\eta.
\]
\end{definition}

The following example illustrates both the separating power and the limitations of observational topology.

\begin{example}[Functional graph bands versus a latent circular mechanism]
\label{ex:functional-latent-circle}
Let the observed variables be \((X,Y)\in\mathbb R^2\), and
define \(T_{\mathrm{obs}}(P^{\mathrm{obs}}) = \beta_1\bigl(\operatorname{supp}P^{\mathrm{obs}}\bigr)\), with homology over a fixed field. We regard this summary as
taking values in \(\mathbb N_0\), equipped with the metric
\(d_{\mathcal B}(m,n)=|m-n|\). Fix \(a<b\) and
\(\delta>0\), and consider the following model families:

\begin{itemize}
\item \(\mathcal F_\to\): \(X\) has support \([a,b]\),
\(\varepsilon\perp X\) has support \([-\delta,\delta]\), and
\(Y=f(X)+\varepsilon\) for a continuous function \(f\);

\item \(\mathcal F_\leftarrow\): \(Y\) has support
\([a,b]\), \(\varepsilon'\perp Y\) has support
\([-\delta,\delta]\), and
\(X=g(Y)+\varepsilon'\) for a continuous function \(g\);

\item \(\mathcal F_\circ: U\sim\operatorname{Unif}[0,2\pi),
\ E=(E_X,E_Y)\perp U\), where \(E\) has support equal to the closed disk of radius \(\delta<r\), and
\(X=r\cos U+E_X, \ Y=r\sin U+E_Y\).
\end{itemize}

The notation \(\mathcal F_\circ\) refers only to the circular
geometry of the observational support. Causally, this is an
acyclic latent-variable model. In an expanded representation,
both \(U\) and \(E\) are latent common parents:
\(U\longrightarrow X, \ \ U\longrightarrow Y, \ \ E\longrightarrow X, \ \ E\longrightarrow Y\). See Figure \ref{fig:dag-topological-separation} for the visual representation of this example.

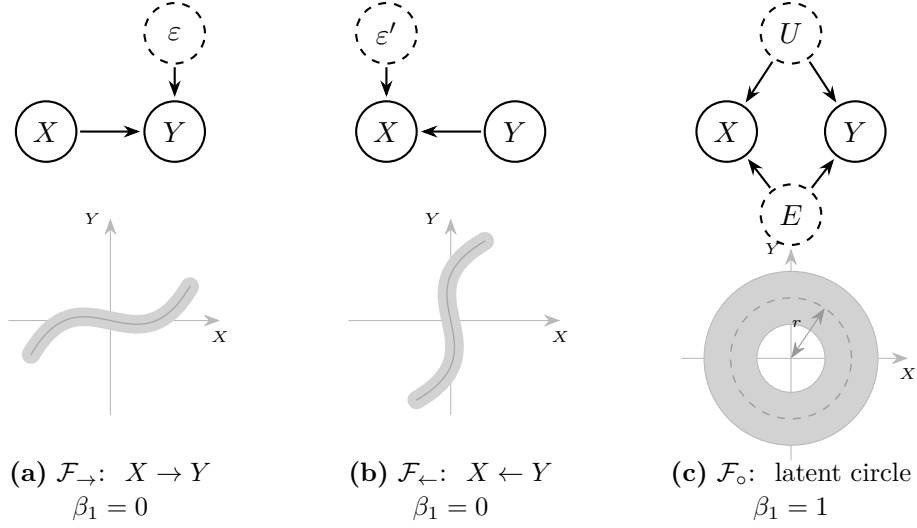
\begin{figure}[ht]
\centering
\begin{tikzpicture}[
    >={Stealth[length=2.2mm,width=1.6mm]},
    obs/.style={circle, draw=black, thick,
                minimum size=8mm, inner sep=0pt},
    lat/.style={circle, draw=black, dashed, thick,
                minimum size=8mm, inner sep=0pt},
    ed/.style={->, thick, shorten <=1pt, shorten >=1.2pt},
    lab/.style={font=\small, align=center},
    ax/.style={->, gray!55, line width=0.4pt}
]
 
%% ============================ (a) F_to ============================
\begin{scope}[shift={(0,0)}]
  \node[obs] (X1) at (0,0)     {$X$};
  \node[obs] (Y1) at (1.7,0)   {$Y$};
  \node[lat] (e1) at (1.7,1.3) {$\varepsilon$};
  \draw[ed] (X1) -- (Y1);
  \draw[ed] (e1) -- (Y1);
 
  \begin{scope}[shift={(0.85,-2.5)}]   % support: band about graph f
    \draw[ax] (-1.35,0) -- (1.45,0) node[below, font=\tiny, black] {$X$};
    \draw[ax] (0,-1.25) -- (0,1.35) node[left,  font=\tiny, black] {$Y$};
    \draw[line width=7pt, gray!35, line cap=round]
         (-1.05,-0.45) .. controls (-0.35,0.75) and (0.35,-0.75) .. (1.05,0.45);
    \draw[line width=0.5pt, gray!70, line cap=round]
         (-1.05,-0.45) .. controls (-0.35,0.75) and (0.35,-0.75) .. (1.05,0.45);
  \end{scope}
 
  \node[lab] at (0.85,-4.75)
    {\textbf{(a)} $\mathcal F_\to$:\; $X\to Y$\\[1pt] $\beta_1=0$};
\end{scope}
 
%% ============================ (b) F_leftarrow ============================
\begin{scope}[shift={(4.5,0)}]
  \node[obs] (X2) at (0,0)     {$X$};
  \node[obs] (Y2) at (1.7,0)   {$Y$};
  \node[lat] (e2) at (0,1.3)   {$\varepsilon'$};
  \draw[ed] (Y2) -- (X2);
  \draw[ed] (e2) -- (X2);
 
  \begin{scope}[shift={(0.85,-2.5)}]   % support: band about graph g
    \draw[ax] (-1.35,0) -- (1.45,0) node[below, font=\tiny, black] {$X$};
    \draw[ax] (0,-1.25) -- (0,1.35) node[left,  font=\tiny, black] {$Y$};
    \draw[line width=7pt, gray!35, line cap=round]
         (-0.45,-1.05) .. controls (0.75,-0.35) and (-0.75,0.35) .. (0.45,1.05);
    \draw[line width=0.5pt, gray!70, line cap=round]
         (-0.45,-1.05) .. controls (0.75,-0.35) and (-0.75,0.35) .. (0.45,1.05);
  \end{scope}
 
  \node[lab] at (0.85,-4.75)
    {\textbf{(b)} $\mathcal F_\leftarrow$:\; $X\leftarrow Y$\\[1pt] $\beta_1=0$};
\end{scope}
 
%% ============================ (c) F_circ ============================
\begin{scope}[shift={(9.0,0)}]
  \node[obs] (X3) at (0,0)      {$X$};
  \node[obs] (Y3) at (1.7,0)    {$Y$};
  \node[lat] (U)  at (0.85,1.3) {$U$};
  \node[lat] (E)  at (0.85,-1.1){$E$};
  \draw[ed] (U) -- (X3);
  \draw[ed] (U) -- (Y3);
  \draw[ed] (E) -- (X3);
  \draw[ed] (E) -- (Y3);
  \begin{scope}[shift={(0.85,-3.0)}]   % support: annulus
    \draw[ax] (-1.45,0) -- (1.55,0) node[below, font=\tiny, black] {$X$};
    \draw[ax] (0,-1.35) -- (0,1.45) node[left,  font=\tiny, black] {$Y$};
    \fill[gray!35, even odd rule] (0,0) circle (1.15) (0,0) circle (0.45);
    \draw[gray!55, line width=0.4pt] (0,0) circle (1.15);
    \draw[gray!55, line width=0.4pt] (0,0) circle (0.45);
    \draw[dashed, gray!80, line width=0.5pt] (0,0) circle (0.8);
    \draw[<->, gray!80, line width=0.4pt] (0,0) -- (56:0.8)
         node[midway, above left=-2pt, font=\tiny, black] {$r$};
  \end{scope}
  \node[lab] at (0.85,-4.75)
    {\textbf{(c)} $\mathcal F_\circ$:\; latent circle\\[1pt] $\beta_1=1$};
\end{scope}
\end{tikzpicture}
\caption{Causal structure (top) and observational support (bottom) for the three model families of Example~\ref{ex:functional-latent-circle}. Dashed
nodes are unobserved; solid nodes are observed. In (c) the exogenous noise $E=(E_X,E_Y)$ is a single \emph{bivariate} latent parent of both $X$ and $Y$: its support is the closed disk of radius $\delta$, so $E_X$ and $E_Y$
are dependent and must not be drawn as two separate noise nodes. The supports in (a) and (b) are homeomorphic to a rectangle, hence $\beta_1=0$, so $\mathcal F_\to$ and $\mathcal F_\leftarrow$ are not separated by $T_{\mathrm{obs}}$ even though they are causally distinct; the
support in (c) is the annulus of radii $r\pm\delta$, hence $\beta_1=1$, and $\mathcal F_\circ$ is topologically separated from both with margin $\eta=1$.}
\label{fig:dag-topological-separation}
\end{figure}

\end{example}

\begin{proposition}[Topological separation and its limit]
\label{prop:functional-circle-separation}
For the model families of Example~\ref{ex:functional-latent-circle},
\[
T_{\mathrm{obs}}(P_M^{\mathrm{obs}})
=
\begin{cases}
0,
&
M\in\mathcal F_\to\cup\mathcal F_\leftarrow,
\\
1,
&
M\in\mathcal F_\circ.
\end{cases}
\]
Consequently, \(T_{\mathrm{obs}}\) identifies the coarser
partition with classes
\(\mathcal F_\to\cup\mathcal F_\leftarrow\) and
\(\mathcal F_\circ\), with margin \(1\), but it does not
identify the finer partition that separates
\(\mathcal F_\to\) from \(\mathcal F_\leftarrow\).
\end{proposition}

\begin{proof}
For \(M\in\mathcal F_\to\), independence gives
\(\operatorname{supp}\mathcal L(X,\varepsilon) = [a,b]\times[-\delta,\delta]\). The continuous map \((x,e)\mapsto(x,f(x)+e)\) sends this compact set onto the
support of \((X,Y)\). Hence
\[
\operatorname{supp}P_M^{\mathrm{obs}}
=
\{(x,y):x\in[a,b],\ |y-f(x)|\leq\delta\}.
\]
This graph band deformation retracts onto the graph of \(f\) through
\[
H_s(x,y)
=
(x,(1-s)y+s f(x)).
\]
The homotopy remains in the band because
\[
|(1-s)y+s f(x)-f(x)|
=
(1-s)|y-f(x)|
\leq\delta.
\]
The graph is homeomorphic to \([a,b]\), so the support is contractible and has \(\beta_1=0\). The reverse family is identical after exchanging the coordinates.

Since \(u\mapsto r(\cos u,\sin u)\) maps the support of \(U\) onto \(C_r\), independence and compactness give
\(\operatorname{supp}P_M^{\mathrm{obs}} =C_r\oplus D_\delta\). Applying the continuous structural map shows that the observational support is the Minkowski sum of
the circle \(C_r=\{z:\|z\|_2=r\}\) and the disk \(D_\delta\):
\[
\operatorname{supp}P_M^{\mathrm{obs}}
=
C_r\oplus D_\delta
=
\{z:r-\delta\leq\|z\|_2\leq r+\delta\}.
\]
Since \(\delta<r\), this is an annulus. The radial homotopy
\[
H_s(z)
=
\left(
(1-s)+s\frac{r}{\|z\|_2}
\right)z
\]
deformation retracts the annulus onto \(C_r\simeq S^1\), so its first Betti number is \(1\). The stated separation and the failure to orient the two functional models follow.
\end{proof}

The separation margin converts stable topological estimation into restricted model-class recovery.

\begin{proposition}[Separation implies robust distinguishability]
\label{prop:stability-separation-discovery}
Suppose the classes in \(\mathcal C/{\sim}\) are pairwise
separated by a common margin \(\eta>0\). Assume that
\(\mathcal C/{\sim}\) is finite or, more generally, that a
measurable nearest-class selection with a fixed tie-breaking
rule is available.

Let \(M\in A\) be the true model and let
\(\widehat\theta_n\in\mathcal B\) estimate \(\theta_T(M)\).
For \(x\in\mathcal B\) and \(S\subseteq\mathcal B\), write
\(d_{\mathcal B}(x,S) = \inf_{\theta\in S}d_{\mathcal B}(x,\theta)\). Define
\(\widehat A_n \in \arg\min_{B\in\mathcal C/{\sim}}
d_{\mathcal B}\bigl(\widehat\theta_n,\Theta_T(B)\bigr)\). If
\(d_{\mathcal B}(\widehat\theta_n,\theta_T(M)) < \frac{\eta}{2}\), then the minimizer is unique and
\(\widehat A_n=A\). Consequently, if \(d_{\mathcal B}(\widehat\theta_n,\theta_T(M)) = o_p(1)\), then
\(\mathbb P(\widehat A_n=A)\longrightarrow1\).
\end{proposition}

\begin{proof}
Set \(e_n = d_{\mathcal B}(\widehat\theta_n,\theta_T(M))\). Because \(\theta_T(M)\in\Theta_T(A)\),
\(d_{\mathcal B}\bigl(\widehat\theta_n,\Theta_T(A)\bigr)
\leq e_n\). For every \(B\neq A\) and every
\(\theta'\in\Theta_T(B)\), class separation gives
\(d_{\mathcal B}(\theta_T(M),\theta')\geq\eta\). Therefore, by the triangle inequality,
\(d_{\mathcal B}(\widehat\theta_n,\theta') \geq
\eta-e_n\). Taking the infimum over \(\theta'\in\Theta_T(B)\) yields \(d_{\mathcal B}\bigl(\widehat\theta_n,\Theta_T(B)\bigr)
\geq \eta-e_n\). If \(e_n<\eta/2\), then
\(d_{\mathcal B}\bigl(\widehat\theta_n,\Theta_T(A)\bigr) <
\frac{\eta}{2} < d_{\mathcal B}\bigl(\widehat\theta_n,\Theta_T(B)\bigr)\) for every \(B\neq A\). Hence \(A\) is the unique nearest
class. Finally, \(e_n=o_p(1)\) implies
\(\mathbb P(e_n<\eta/2)\to1\), which proves the consistency
claim.
\end{proof}

This result identifies the appropriate scope of topology-assisted causal discovery. Topology can discriminate among restricted mechanism classes when their attainable summaries are separated by more than the estimation error. It does not, without additional assumptions, orient causal edges or identify interventions. Such claims still require
information such as conditional independences, independent-noise asymmetry, temporal order, interventions, or invariance across environments~\citep{spirtes2000,pearl2009}.

\begin{remark}
The support-based summary in
Example~\ref{ex:functional-latent-circle} is intended as a transparent population illustration. Literal support topology may be unstable under contamination or unbounded noise. Practical implementations would require a stable density, high-probability-region, or distance-to-measure construction together with a consistency result for its empirical estimator.
\end{remark}

%%%%%%%%%%%%%%%%%%%%%%%%%%%%%%%%%%%%%%%%%%%%%%%%%%%%

\section{Identification under topological ignorability}
\label{sec:topological-ignorability}

All previous identification results were stated under joint conditional exchangeability, \((Y^0,Y^1)\perp A\mid X\). For identification of treatment-specific marginal laws, the weaker arm-specific condition
\[
Y^a\perp A\mid X,
\qquad
a\in\{0,1\},
\]
is sufficient; we call this \emph{weak conditional exchangeability}.

For a fixed, possibly non-injective topological representation, the covariate-standardized effect \(\tau_{\mathrm{dist}}\) may instead be identified under the still more target-specific condition of \emph{topological ignorability}, introduced and developed by Saki and Faghihi~\citep{sakifaghihi2026}. This is not a general robustness guarantee against hidden confounding: it is an alternative identifying assumption concerning only the selected topological summary.

Fix versions of the relevant regular conditional laws and write
\[
P_{a,x}
:=
\mathcal L(Y^a\mid X=x),
\qquad
P^{\mathrm{fact}}_{a,x}
:=
\mathcal L(Y^a\mid A=a,X=x).
\]
The first is the treatment-\(a\) potential-outcome law in stratum \(x\); the second is its law among units factually assigned to treatment \(a\). Recall also the observed conditional kernel
\[
Q_a(x,\cdot)
=
\mathcal L(Y\mid A=a,X=x).
\]
Consistency gives \(P^{\mathrm{fact}}_{a,x} = Q_a(x,\cdot)\) for \(P_{X\mid A=a}\)-almost every \(x\). Under strict positivity, this equality holds for \(P_X\)-almost every \(x\).

\begin{definition}[Conditional topological ignorability]
\label{def:topological-ignorability}
Let
\(T_{\mathrm{dist}}: \mathcal P_\star(\mathcal Y)\to\mathcal B
\) be fixed, and suppose that \(P_{a,x}\) and
\(P^{\mathrm{fact}}_{a,x}\) belong to its domain for
\(a=0,1\) and \(P_X\)-almost every \(x\).
Conditional topological ignorability relative to
\(T_{\mathrm{dist}}\) holds if
\[
T_{\mathrm{dist}}(P_{a,x})
=
T_{\mathrm{dist}}(P^{\mathrm{fact}}_{a,x}),
\qquad
a\in\{0,1\},
\]
for \(P_X\)-almost every \(x\).
\end{definition}

The following proposition restates results of Saki and
Faghihi~\citep{sakifaghihi2026}.

\begin{proposition}[Identification under topological ignorability]
\label{prop:tau-dist-ignorability}
Assume consistency and strict positivity.

\begin{enumerate}
\item Weak conditional exchangeability implies conditional topological ignorability for every \(T_{\mathrm{dist}}\). Hence the joint exchangeability assumption used in the preceding sections also implies conditional topological ignorability.

\item Suppose treatment is binary and \(T_{\mathrm{dist}}\) is injective on a class containing \(P_{a,x}\) and
\(P^{\mathrm{fact}}_{a,x}\) for \(a=0,1\) and
\(P_X\)-almost every \(x\). Then conditional topological ignorability is equivalent to weak conditional exchangeability: \(Y^a\perp A\mid X, \ a\in\{0,1\}\).

\item Suppose conditional topological ignorability holds and
\(x\longmapsto \left\| T_{\mathrm{dist}}(P_{1,x})
- T_{\mathrm{dist}}(P_{0,x}) \right\|_{\mathcal B}
\) is measurable and integrable. Then
\[
\tau_{\mathrm{dist}}
=
\mathbb E_X\!\left[
\left\|
T_{\mathrm{dist}}\bigl(Q_1(X,\cdot)\bigr)
-
T_{\mathrm{dist}}\bigl(Q_0(X,\cdot)\bigr)
\right\|_{\mathcal B}
\right].
\]
Consequently, \(\tau_{\mathrm{dist}}\) is identified from the observational law.
\end{enumerate}
\end{proposition}

Non-injectivity is necessary but not sufficient for conditional topological ignorability to be strictly weaker than weak exchangeability: the model class must contain distinct relevant laws in the same fiber of \(T_{\mathrm{dist}}\).

This result identifies only the covariate-standardized within-stratum contrast. In general, it does not identify the marginal effect
\(\Delta_{\mathrm{dist}} = T_{\mathrm{dist}}\bigl(\mathcal L(Y^1)\bigr) - T_{\mathrm{dist}}\bigl(\mathcal L(Y^0)\bigr)\), because \(T_{\mathrm{dist}}\) need not commute with mixing over \(X\); see Remark~\ref{rem:mix-first} and
Theorem~\ref{thm:noncommutation-characterization}.

\subsection{Structural interpretation}
\label{subsec:filtered-chambers}

Saki and Faghihi~\citep{sakifaghihi2026} give structural sufficient conditions for topological ignorability through latent reweighting. Suppose that an unobserved variable \(U\) suffices to control treatment selection:
\(Y^a\perp A\mid(X,U)\). Write
\[
\kappa_{a,x,u}
:=
\mathcal L(Y^a\mid X=x,U=u),
\qquad
\Pi_x
:=
\mathcal L(U\mid X=x),
\]
and
\(\Pi^{\mathrm{obs}}_{a,x} := \mathcal L(U\mid A=a,X=x)\). Then, for \(P_X\)-almost every \(x\),
\[
P_{a,x}
=
\int
\kappa_{a,x,u}\,d\Pi_x(u),
\qquad
P^{\mathrm{fact}}_{a,x}
=
\int
\kappa_{a,x,u}\,
d\Pi^{\mathrm{obs}}_{a,x}(u).
\]

Thus hidden confounding changes the mixing law, although both mixtures use the same conditional kernels.

Assume additionally that the relevant conditional laws admit specified density representatives with respect to a common reference measure, so that the superlevel-set construction is a well-defined functional of the law. For a finite threshold set \(\Lambda\), suppose the interventional and factual superlevel filtrations are related by homotopy equivalences that are compatible, up to homotopy, with the inclusions across thresholds. The induced persistence modules over \(\Lambda\) are then isomorphic. Consequently, any $T_{dist}$ that factors through these modules takes the same value on \(P_{a,x}\) and \(P^{\mathrm{fact}}_{a,x}\), which implies conditional topological ignorability~\citep{sakifaghihi2026}.

A stronger sufficient condition applies when the latent-regime densities are piecewise constant on a fixed finite cell complex: no cell may change its superlevel-set membership at any threshold in \(\Lambda\) under the latent reweighting. This gives identical finite filtrations. If, in addition,
\(U\in\{1,\ldots,r\}\), the cell values are affine functions of the mixing weights, and the threshold equalities define hyperplanes in the finite-dimensional probability simplex. Their sign patterns partition the simplex into polyhedral cells. The no-switching criterion is sufficient for the two mixing laws to lie in the same filtered topological chamber, but it is not necessary for their persistence modules or summaries to agree~\citep{sakifaghihi2026}.

Topological ignorability remains an untestable,
representation-specific causal assumption. It permits identification of the selected coarse topological target without identifying the full interventional laws or the ordinary mean effect. It should therefore be accompanied by a clear scientific justification and, where possible, sensitivity analysis.

%%%%%%%%%%%%%%%%%%%%%%%%%%%%%%%%%%%%%%%%%%%%%%%%%%%%

\section{Illustrative examples}
\label{sec:minimal-examples}

We briefly illustrate the three roles that topology may play in the framework: as a representation of individual potential outcomes, as a diagnostic of residual structure, and as a representation of interventional laws.

\paragraph{Shape-valued treatment effects.}
Suppose that \(Y^a\) is a voxelized or triangulated tumour shape in \(\mathbb R^3\) (see Figure \ref{fig:shape-tate}). For a chosen filtration \(\mathcal F\), set
\(Z^a = \Phi\bigl( \operatorname{Dgm}_k(\mathcal F(Y^a))
\bigr)\), where \(\Phi\) takes values in a Banach space \(\mathcal B\), and assume that \(Z^0\) and \(Z^1\) are Bochner integrable. The outcome-level effect is
\[
\operatorname{TATE}_{\Phi,k}
=
\mathbb E[Z^1]-\mathbb E[Z^0]
\in\mathcal B.
\]
For \(k=0,1,2\), it may respectively capture changes in fragmentation, tunnel-like structure, or enclosed cavities. If
\(\mathcal B=C(I)\) for a compact filtration interval \(I\), the effect is a continuous curve over filtration scales.

\begin{figure}[ht]
    \centering
    \includegraphics[width=\textwidth]
    {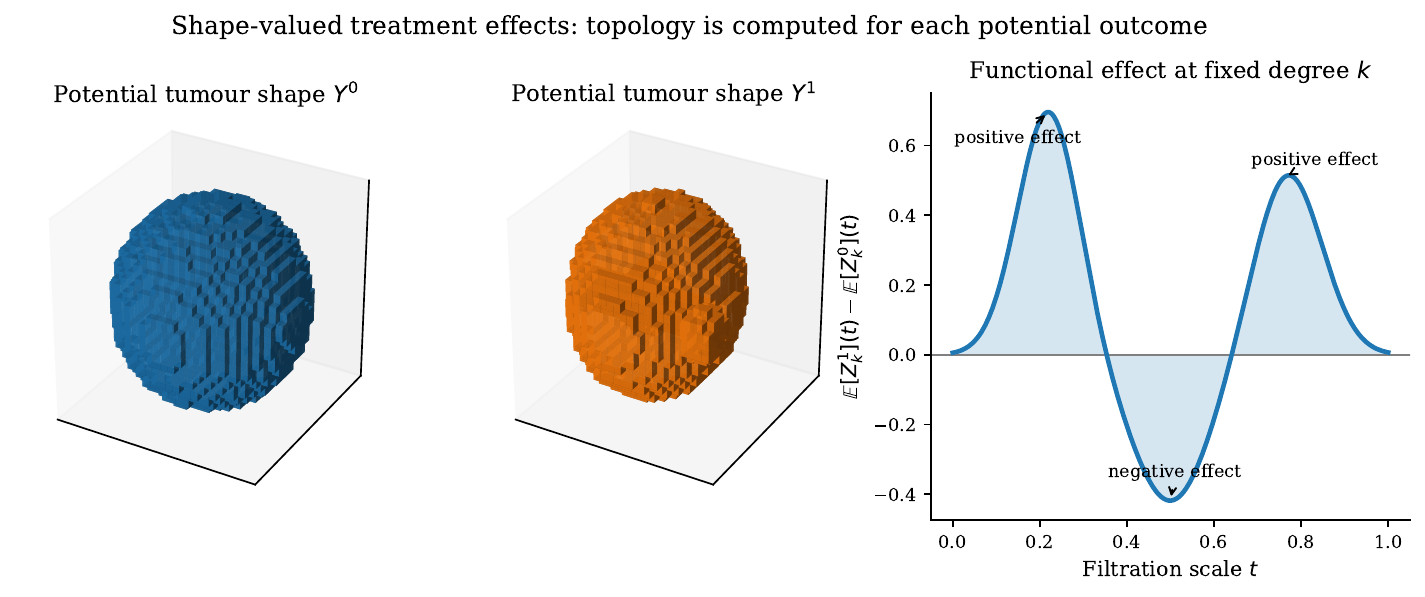}
    \caption{
    Illustration of a shape-valued treatment effect.
    The potential outcomes \(Y^0\) and \(Y^1\) are voxelized tumour shapes whose topological representations may differ in their connected components, tunnel-like structures, and enclosed cavities. When \(\mathcal B=C(I)\), their mean difference is a function of the filtration scale.
    }
    \label{fig:shape-tate}
\end{figure}

\paragraph{Residual topology as a diagnostic.}
Let \(U\sim\operatorname{Unif}(0,2\pi)\) and consider
\[
V=\cos U+\varepsilon_V,
\qquad
W=\sin U+\varepsilon_W,
\]
where the noise variables are centered, mutually independent, and independent of \(U\) as shown in Figure \ref{fig:residual-topology}. Then \(\operatorname{Cov}(V,W)=0\), although for sufficiently small noise the empirical cloud \(\{(V_i,W_i)\}\) may exhibit a persistent \(H_1\) feature. Likewise, persistent homology of a cloud formed from covariates and residuals after fitting a working model may reveal nonlinear structure or model misspecification.

This use is diagnostic only: persistent residual topology does not by itself identify a confounder, establish causal direction, or distinguish confounding from other forms of misspecification.

\begin{figure}[ht]
    \centering
    \includegraphics[width=\textwidth]
    {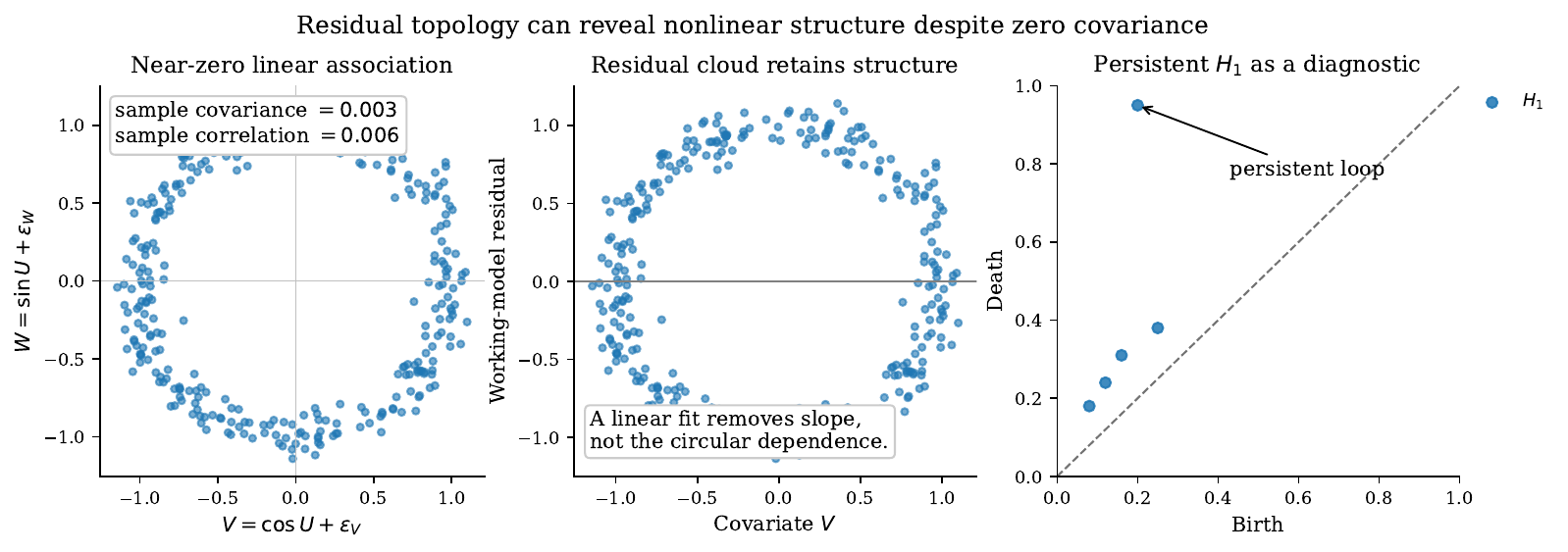}
    \caption{
    Residual topology as a model diagnostic.
    Although \(V\) and \(W\) have zero population covariance, their joint cloud has circular structure. Fitting a linear working model removes linear association but does not remove this nonlinear dependence, which appears as a persistent \(H_1\) feature. Such persistence indicates residual structure but does not, by itself, determine its causal origin.
    }
    \label{fig:residual-topology}
\end{figure}

\paragraph{Distribution-level treatment effects.}
Suppose that \(Y^a\in\mathbb R^d\) has finite second moment and that the scientific target is the geometry of its population law \(P_Y^a=\mathcal L(Y^a)\). For fixed \(m\in(0,1)\), let \(T_{\mathrm{dist}}(P_Y^a) = \Phi\bigl(
D_{k,m}^{\mathrm{DTM}}(P_Y^a) \bigr)\), where \(\Phi\) takes values in a Banach space. The resulting distribution-level effect is
\[
\Delta_{\mathrm{dist}}
=
T_{\mathrm{dist}}(P_Y^1)
-
T_{\mathrm{dist}}(P_Y^0).
\]
It can distinguish, for example, one persistent cluster from two separated persistent clusters even when
\(\mathbb E[Y^1]=\mathbb E[Y^0]\). Figure \ref{fig:distribution-effect} illustrates this example. Unlike the first example, topology is applied after forming the interventional law rather than separately to each potential outcome.

\begin{figure}[ht]
    \centering
    \includegraphics[width=\textwidth]
    {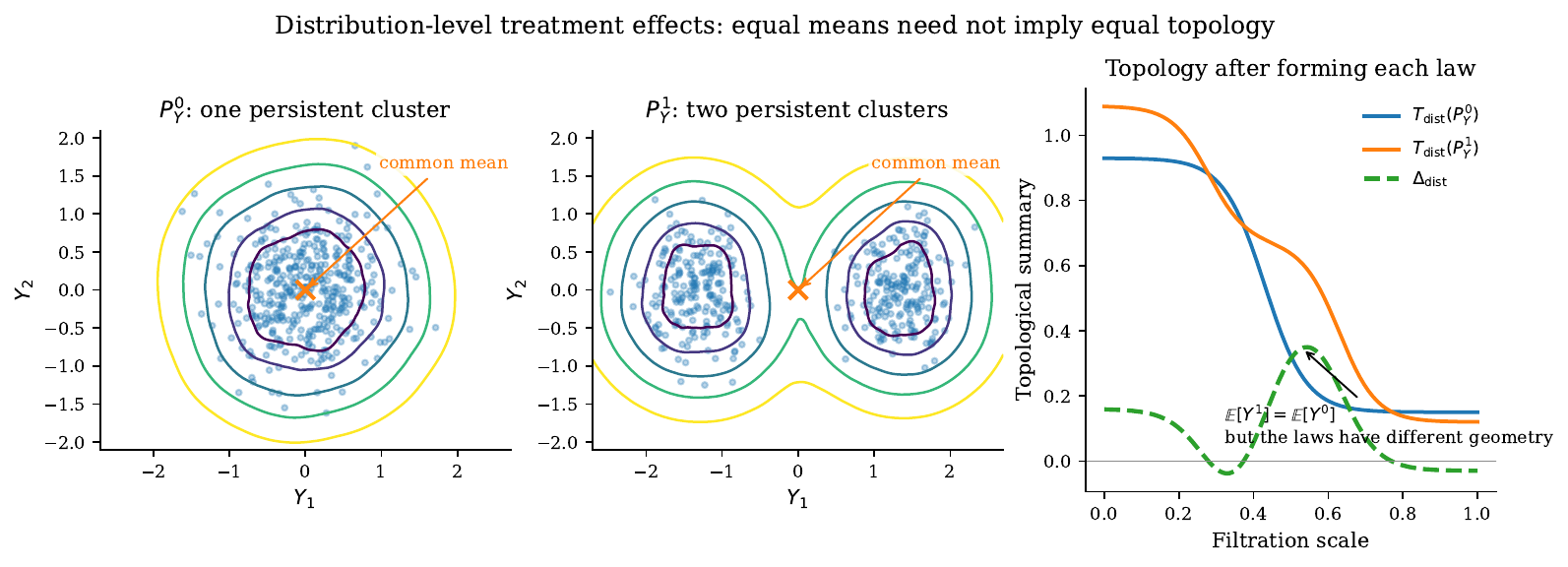}
    \caption{
    Illustration of a distribution-level treatment effect.
    The two interventional laws have the same mean but different geometric structure: \(P_Y^0\) contains one persistent cluster, whereas \(P_Y^1\) contains two separated persistent clusters. The contours show empirical distance-to-measure approximations.
    Here topology is applied to each interventional law after that law has been formed, rather than separately to individual potential outcomes.
    }
    \label{fig:distribution-effect}
\end{figure}

%%%%%%%%%%%%%%%%%%%%%%%%%%%%%%%%%%%%%%%%%%%%%%%%%%%%%

\section{Conclusion}
\label{sec:conclusion}

This paper formulated Topological Causal Data Analysis as a four-layer framework consisting of an observation space, a causal-model class, a topological representation, and a causal query. Keeping these layers separate makes clear that topology neither defines interventions nor replaces causal assumptions. Its role is to provide stable, shape-sensitive representations of structured outcomes and probability laws.

The framework distinguishes outcome-level TCDA, which averages topological representations of individual potential outcomes, from distribution-level TCDA, which applies topology to interventional outcome laws. These operations generally do not commute and therefore define different causal targets. The distinction matters when treatment changes clustering, connectivity, loops, cavities, or other population-level geometry without substantially changing ordinary mean outcomes.

At the outcome level, standard causal methods apply after the structured outcome is mapped into a suitable Banach space. The power-weighted silhouette estimand, its efficient influence function, functional weak convergence, testing procedure, and silhouette-stability results are due to Kim and Lee~\citep{kimlee2026}; vectorization-free Fr\'echet contrasts draw on the metric-space theory of Shin et al.~\citep{shinetal2024}. At the distribution level, causal identification must precede the topological transformation. The framework makes this order explicit, characterizes when outcome- and distribution-level constructions agree, and transfers stability and interventional-law estimation error to the resulting topological effects.

Following Saki and Faghihi~\citep{sakifaghihi2026}, we also considered conditional topological ignorability. On appropriate model classes with distinct laws in the same fiber of a non-injective representation, this target-specific assumption can be weaker than weak conditional exchangeability while still identifying the covariate-standardized topological effect. It does not generally identify the marginal interventional topology, the full interventional laws, or ordinary mean effects.

Finally, observational topology may reveal nonlinear dependence, residual structure, or latent geometry, but it cannot by itself orient causal relations or recover a causal graph. TCDA therefore places topology within the usual causal sequence of target definition, assumptions, identification, estimation, and sensitivity analysis. Future work may develop sharper statistical theory, computational tools, and applications involving images, shapes, networks, spatial fields, and other structured outcomes.
%%%%%%%%%%%%%%%%%%%%%%%%%%%%%%%%%%%%

\bibliographystyle{plainnat}
\bibliography{references}

\end{document}